\documentclass{article}
\usepackage[margin=1.25in]{geometry}
\usepackage{amsmath}
\usepackage{graphicx}
\usepackage{amsthm}
\usepackage{amssymb}
\usepackage{lineno}
\usepackage[caption=false]{subfig}
\usepackage{tikz}
\usepackage{hyperref}

\newtheorem{lemma}{Lemma}[section]

\newtheorem{corollary}[lemma]{Corollary}
\DeclareMathOperator{\diag}{diag}
\tikzstyle{stage}=[circle,text centered,draw=black,fill=white]
\tikzstyle{arrow} = [thick,->,>=stealth]
\tikzstyle{stage2}=[circle,text centered,draw=black,fill=white, minimum size=1cm]

\date{\today}

\title{A Framework for Studying Transients in Marine Metapopulations\thanks{This work was supported by the Natural Science and Engineering Research Council of Canada (NSERC) Discovery Grants (MAL and PvdD).}}

\author{Peter D. Harrington\thanks{Department of Mathematical and Statistical Sciences, University of Alberta, Edmonton, Canada. \newline (\href{mailto:harringt@ualberta.ca}{harringt@ualberta.ca},  \href{mailto:mark.lewis@ualberta.ca}{mark.lewis@ualberta.ca})},\,  Mark A. Lewis\footnotemark[2] \thanks{Department of Biological Sciences, University of Alberta, Edmonton, Canada.},\, and P. van den Driessche\thanks{Department of Mathematics and Statistics, University of Victoria, Victoria, Canada. (\href{mailto:vandendr@uvic.ca}{vandendr@uvic.ca})}}   

\begin{document}

\maketitle

\begin{abstract}
Transient dynamics can often differ drastically from the asymptotic dynamics of systems. In this paper we provide a unifying framework for analysing transient dynamics in marine metapopulations, from the choice of norms to the addition of stage structure. We use the $\ell_1$ norm, because of its biological interpretation, to extend the transient metrics of reactivity and attenuation to marine metapopulations,  and use examples to compare these metrics under the more commonly used $\ell_2$ norm. We then connect the reactivity and attenuation of marine metapopulations to the source-sink distribution of habitat patches and demonstrate how to meaningfully measure reactivity when metapopulations are stage-structured. 
\end{abstract}



\section{Introduction}
\label{sec:intro}

Transient dynamics, those that occur over short timescales, can often be vastly different from the asymptotic or long term dynamics of ecological systems. However, throughout the history of mathematical biology much of the work has focused on determining the asymptotic dynamics of biological systems. While the study of long-term dynamics has given ecologists many tools to analyze the behaviour of populations, these tools are often not the same as those required to understand transient dynamics. Recently Hastings et al. \cite{Hastings2018} have shown that transient dynamics are much more ubiquitous than previously assumed and long transients occur in many different ecological systems, from plankton and coral to voles and grouse. Studying the transient dynamics of an ecological system can give useful insight into the different processes that may occur after a disturbance, change in environmental conditions, or change in human intervention to a system. In some marine systems that are driven by environmental fluctuations, such as the Dungeness crab, transient dynamics may in fact be key to understanding how these systems behave \cite{Higgins1997}. 

There has also been a recent push to characterize the different types of systems which display long transient dynamics that differ significantly from their asymptotic dynamics \cite{Morozov2020, Hastings2018, Hastings2004, Hastings2001}. Hastings et al. \cite{Hastings2018} have loosely categorized four different drivers of long transient dynamics in ecological systems: ghost attractors and crawl-bys, slow-fast dynamics,  high dimensionality, and stochastic noise. These categories are not always distinct and certain systems may indeed fall into multiple categories. For example, a predator prey system may have a crawl-by past a saddle node which drives the transient dynamics in this system, but this could also be thought of as a difference in timescales of the predator decline due to lack of prey. For metapopulations the main driver of transient dynamics is often the high dimensionality arising due to spatial structure, though often these transient dynamics are exacerbated by the other drivers as well.

Some of the earliest studies of systems that could generate long transients were systems with spatial structure \cite{Hastings1994, lloyd1996}. It seems intuitive that spatial structure or spatial heterogeneity can drive some sort of transient dynamics in a system.  If individuals start in one location in a habitat, especially a poor habitat, then it will take time before they can spread over the entire habitat and the long-term population dynamics begin to emerge. What is surprising is that spatial structure can also give rise to so called long-lived transients, where the transient dynamics are extensive enough that they continue on timescales past which we typically measure biological populations \cite{Hastings1994}. 

One method of adding spatial structure to a population is to formulate it as a metapopulation, where distinct populations live on habitat patches which are connected via dispersal or migration. Metapopulation models were originally proposed by Levins  \cite{Levins1969} to model patch occupancy in habitats consisting of isolated habitat patches, but these early models used space implicitly rather than explicitly. Later metapopulation models have included space explicitly by allowing for differing habitat quality on patches or differing dispersal between patches \cite{Hanski1994, Gyllenberg1997b}, though often these models are focused on the proportion of occupied patches rather than the population size on each patch. However, many marine metapopulation models as well as epidemiological metapopulation models explicitly track the number of individuals on each patch as well as movement or dispersal between patches \cite{lloyd1996, Arino2003, Figueira2006, Armsworth2002}. In this paper we model the metapopulation structure following this spatially explicit framework where individuals are tracked rather than the proportion of occupied patches.

Another benefit of the metapopulation framework is that habitat patches can be classified into source patches and sink patches. This classification can occur in many different ways \cite{Figueira2006, Pulliam1988, Krkosek2010a}, but commonly a source is a productive habitat patch and a sink is a poor habitat patch. Early measures of sources and sinks were mainly focused on connectivity between patches, however more recently it has been understood that it is the interplay between patch connectivity and local patch productivity that characterizes patches as sources or sinks. One of the new and easily tractable metrics that embodies this relationship comes from the theory of next-generation matrices and the basic reproduction number, $R_0$. This framework, originally developed in epidemiology, has been used to characterize sources and sinks in populations of mussels, salmon, and sea lice on salmon farms \cite{Krkosek2010a, Huang2015, Harrington2020}. 

While metapopulation theory has previously been used to classify patches as sources and sinks, other metrics have been used to characterize the transient dynamics of systems. Reactivity was initially introduced by Neubert and Caswell \cite{Caswell1997} to measure the maximum initial growth rate of a system over all possible perturbations from an equilibrium. If the maximum initial growth rate is positive, then the system is reactive. Complimenting reactivity is the amplification envelope, which is the maximum possible amplification at time $t$ that can be achieved by a perturbation. Later, Townley and Hodgson \cite{Townley2008} introduced attenuation as the opposite metric to measure initial decline of populations; a system attenuates if the minimum possible growth rate declines following a perturbation. Reactivity and attenuation are then most interesting when they are different from the stability of the equilibrium of a system --- when a system attenuates but is unstable, or is reactive but stable --- and it is in these situations that we focus this paper.

It should be noted that reactivity, attenuation and the amplification envelope are all defined from the linearization of a non-linear system about an equilibrium. These measures are therefore most useful around hyperbolic equilibria, where the dynamics of the non-linear system can be well approximated by the dynamics of the linear system. If an equilibrium is not hyperbolic then the trajectories in the non-linear system may no longer be similar to the linearization by which reactivity, attenuation, and the amplification envelope are defined. Even around a hyperbolic equilibrium the trajectories of the non-linear and linearsized systems may diverge as they move away from the equilibrium. Here we use the technique of linearization to determine reactivity and attenuation as others have before us, but want to emphasize these caveats as they are often brushed over in the transient literature.

In this paper we apply these transient measures of growth to a class of biological metapopulation models where there is no  migration between population patches, only birth on new patches. These are a subset of birth-jump processes \cite{Hillen2015} and include models for marine meroplanktonic species, where larvae can travel through the ocean between population patches but adults remain confined to a habitat patch. Specific species that exhibit this structure include sea lice \cite{Adams2015}, corals and coral reef fish \cite{Cowen2006, Jones2009}, barnacles \cite{Roughgarden1988}, Dungeness crabs \cite{Botsford1994}, sea urchins \cite{Botsford1994}, and many other benthic marine species \cite{Cowen2009}. This type of system also encompasses many plant species where seeds are carried between suitable habitat patches \cite{Husband1996}, and depending on the census timing could also include insect species where there is one large dispersal event between habitat patches, such as the spruce budworm \cite{Ludwig1978, Morris1963, Williams2000} and mountain pine beetle \cite{Safranyik2007}. Lastly this class of models also includes multi-patch or multi-city epidemiological metapopulation models where infections can spread between patches, for example  infected residents of a city may travel and infect residents of other cities before returning home \cite{Arino2003}.

The paper is structured as follows. First we extend the general theory of transients to marine and other birth-jump metapopulations. We then demonstrate that even for linear two-patch metapopulations it is possible to have transient dynamics that are vastly different than the asymptotic dynamics of these populations and moreover that the time for which these transient dynamics occur can extend for an arbitrarily long amount of time. Next, we examine how the structure of the metapopulation can enhance the transient dynamics of metapopulations. Then we demonstrate how to connect the transient behaviour of the metapopulation to its source-sink distribution. Finally, we discuss how to measure the transient dynamics of stage-structured metapopulations in a useful and biological meaningful manner.

\section{Extending the general theory of transients to metapopulations}
\label{sec:gentheory}
In this section we apply the metrics of reactivity \cite{Caswell1997} and attenuation \cite{Townley2008} to general systems of single-species metapopulations. We demonstrate that often reactivity and attenuation in metapopulations have simple closed forms, especially if we are analysing how population size responds to initial perturbations. In order to present our work in a general form, we model the dynamics of a metapopulation of a single species on $n$ patches around the zero equilibrium with the system:
\begin{equation}
x'=Ax,\label{eq:linear}
\end{equation}
where $A=[a_{ij}]$ is a real irreducible matrix of order $n$. This most often represents the linearization of a non-linear system, which more completely captures the dynamics of the population but could also represent the full dynamics of a linear system if density dependence was not important to the population dynamics. 

For the analyses in this paper we focus on biologically realistic single-species metapopulations where the entries of $x(t)$ are non-negative when beginning with a non-negative initial condition, $x(0)$. This condition is equivalent to requiring that $A$ be an essentially non-negative (Metzler matrix), such that all the off-diagonal entries of $A$ are non-negative (Thm 2.4, \cite{Thieme2009}). Biologically this means that the presence of individuals on one patch cannot contribute to the decline of a population on another patch and that the population on each patch will not become negative. 

\subsection{Reactivity and attenuation using the $\ell_1$ norm}

To analyze the transient dynamics of this metapopulation let us introduce some notation from Neubert and Caswell \cite{Caswell1997}. We say that a population is \textit{reactive} if there is an initial condition such that the initial growth rate of the total population is positive. Using notation from Lutscher and Wang \cite{Lutscher2020}, reactivity is formally defined as:
\begin{equation}
\bar{\sigma}=\max_{||x_0||\neq 0} \left[ \frac{1}{||x||}\frac{d||x||}{dt}\biggr\rvert_{t=0}\right].\label{eq:reactive}
\end{equation}
If $\bar{\sigma}>0$ then a system is reactive, and if $\bar{\sigma}\leq 0$ then the system is not reactive. Neubert and Caswell \cite{Caswell1997} use the $\ell_2$ norm to measure the population size, as this allows for the simplification that $\bar{\sigma}$ is the max eigenvalue of $(A+A^T)/2$. However, the $\ell_2$ norm lacks a reasonable biological interpretation, and so others have instead used the $\ell_1$ norm to define reactivity \cite{Huang2015, Townley2007, Stott2011}. Biologically, the $\ell_1$ norm, 
\begin{equation*}
||x||_1= \sum_{i=1}^n |x_i|
\end{equation*}
can be interpreted as the total population on all patches of a metapopulation. Moreover in the metapopulation framework using the $\ell_1$ norm is also convenient to determine reactivity from the population matrix $A$.

Similar to reactivity, we also say that a population \textit{attenuates} if there is an initial condition for which the initial growth rate of the total population declines \cite{Townley2008}. This is formally defined as 
\begin{equation}
\underline{\sigma}=\min_{||x_0||\neq 0} \left[ \frac{1}{||x||}\frac{d||x||}{dt}\biggr\rvert_{t=0}\right].
\end{equation}
If $\underline{\sigma}<0$ then the system attenuates, and if $\underline{\sigma}\geq 0$ then the system does not attenuate. Comparing the definitions of attenuation and reactivity we can see that it is possible for a system to be both reactive and to attenuate, if there are certain initial conditions for which $\bar{\sigma}>0$ is achieved and others such that $\underline{\sigma}<0$. In fact, reactivity and attenuation are most interesting when they are different from the stability of the system, as this is when the transient dynamics are different than the long term dynamics. This is when a system is reactive but stable, so that the total population initially grows but eventually declines, or when a system attenuates but is unstable, so that the total population declines but eventually grows. It should also be noted that the only systems which are not reactive and do not attenuate are those in which the total population size remains constant for all time.

The last measures that we define here to use in some later sections are the \textit{amplification envelope} and the \textit{maximum amplification}. The amplification envelope is the maximum possible amplification at time $t$ from a perturbation $x_0$ and is defined mathematically as:
\begin{equation}
\rho (t)= \max_{||x_0|| \neq 0} \frac{||x(t)||}{||x_0||}.
\label{eq:ampenv} 
\end{equation}
The maximum amplification is simply the maximum of the amplification envelope over all time:
\begin{equation}
\rho_{\max}=\max_{t\geq 0} \rho(t)=\max_{\substack{ t \geq 0 \\ ||x_0||\neq 0}} \frac{||x(t)||}{||x_0||}.
\label{eq:maxamp}
\end{equation}
While reactivity and attenuation quantify the short time response to a perturbation, the amplification envelope and maximum amplification quantify how large a perturbation can become and how long growth can last. It is for these purposes that we use the amplification envelope  and maximum amplification in subsection \ref{sec:largegrowth} and section \ref{sec:largepatch}.

Now before quantifying the reactivity and attenuation of the entire metapopulation, let us first determine the initial growth rate of the population if we begin with one individual on patch $j$. We call this initial growth rate  $\lambda_j$, and mathematically we define 
\begin{equation*}
\lambda_j=\sum_{i=1}^n x_i'(0),
\end{equation*}
with $x(0)=e_j$, where $e_j$ is the vector of length $n$ with $1$ in the $j$th entry and $0$s elsewhere. In terms of system (\ref{eq:linear}) this simplifies to the $j$th column sum of $A$,
\begin{equation*}
\lambda_j=\sum_{i=1}^n a_{ij}.
\end{equation*}
The initial growth rate for a given patch $j$, $\lambda_j$, can also be calculated from the lifecycle graph as the sum of all the outgoing birth rates from a patch minus the death rate on that patch, where any paths describing movement of individuals between patches are ignored. 
  We can then connect this patch specific initial growth rate with the total growth rate, or reactivity, using the following lemma.

\begin{lemma}
Assuming that population sizes are non-negative on each patch, so that $x_j\geq 0$ for all $j$, then for system (\ref{eq:linear}) under the $\ell_1$ norm, 
\begin{equation*}
\bar{\sigma}=\max_{1\leq j\leq n} \lambda_j=\max_{1\leq j\leq n}\sum_i a_{ij}.
\end{equation*}
\label{lemma:reactive}
\end{lemma}
\begin{proof}

Since $x_j\geq 0$ for all $j$, the absolute value signs in (\ref{eq:reactive}) can be dropped and so
\begin{align*}
\bar{\sigma}&=\max_{||x_0||=1} \left[ \frac{d||x||}{dt}\biggr\rvert_{t=0}\right]\\
&=\max_{||x_0||=1} \left[ \frac{d}{dt}\sum_{i=1}^n x_i \biggr\rvert_{t=0}\right]\\
&=\max_{||x_0||=1} \left[ \sum_{i=1}^n \frac{d}{dt}x_i \biggr\rvert_{t=0}\right]\\
&=\max_{||x_0||=1} \left[ 1^T x'\biggr\rvert_{t=0}\right].
\end{align*}
Substituting $x'=Ax$ from system (\ref{eq:linear}) gives 
\begin{align*}
\bar{\sigma}&=\max_{||x_0||=1} \left[ 1^T Ax_0 \right]\\
&=\max_{||x_0||=1} \left[ \sum_{j=1}^n \left(\sum_{i=1}^n a_{ij}\right)x_{0j} \right]
\end{align*}

Now let $k$, with $1\leq k\leq n$ maximize $\sum_{i=1}^n a_{ij}$, so that $\sum_{i=1}^n a_{ik}=\max_{1\leq j\leq n} \sum_{i=1}^n a_{ij}$. Then, under $||x_0||=1$, 

\begin{align*}
\sum_{j=1}^n \left(\sum_{i=1}^n a_{ij}\right)x_{0j} &\leq  \left(\sum_{i=1}^n a_{ik}\right) \sum_{j=1}^n x_{0j}\\
&=\left(\sum_{i=1}^n a_{ik}\right),
\end{align*}
because $||x_0||=1$, with equality when $x_{0j}=\begin{cases}
1 & j=k\\
0 & j\neq k
\end{cases}$. Therefore 

\begin{align*}
\bar{\sigma}&=\max_{1\leq j\leq n} \left(\sum_{i=1}^n a_{ij}\right)\\
&=\max_{1\leq j\leq n} \lambda_j.
\end{align*}
\end{proof}

With a similar proof we can connect the patch specific initial growth rate to attenuation via the following lemma:

\begin{lemma}
Assuming that population sizes are non-negative on each patch, so that $x_j\geq 0$ for all $j$, then for system (\ref{eq:linear}) under the $\ell_1$ norm, 
\begin{equation*}
\underline{\sigma}=\min_{1\leq j\leq n} \lambda_j=\min_{1\leq j\leq n} \sum_i a_{ij}.
\end{equation*}
\end{lemma}

\subsection{Comparing reactivity in the $\ell_1$ and $\ell_2$ norms}
\label{sec:reacl1l2}

In this section we present some examples of systems that are reactive in $\ell_1$ but not in $\ell_2$ and vice versa to illuminate the difference between measuring reactivity in the two norms. It has previously been noted that reactivity depends on the norm in which it is measured \cite{Caswell1997, Lutscher2020}, and the following examples help clarify the underlying biological and mathematical meaning of the two norms. The $\ell_1$ norm,
\begin{equation*}
||x||_1= \sum_{i=1}^n |x_i|,
\end{equation*}
represents the total population size of the metapopulation, whereas the $\ell_2$ norm,
\begin{equation*}
||x||_2= \sqrt{\sum_{i=1}^n x_i^2},
\end{equation*}
represents a sort of Euclidean distance of the metapopulation away from the origin.

\begin{figure}[!htb]
\centering
\includegraphics[width=9cm]{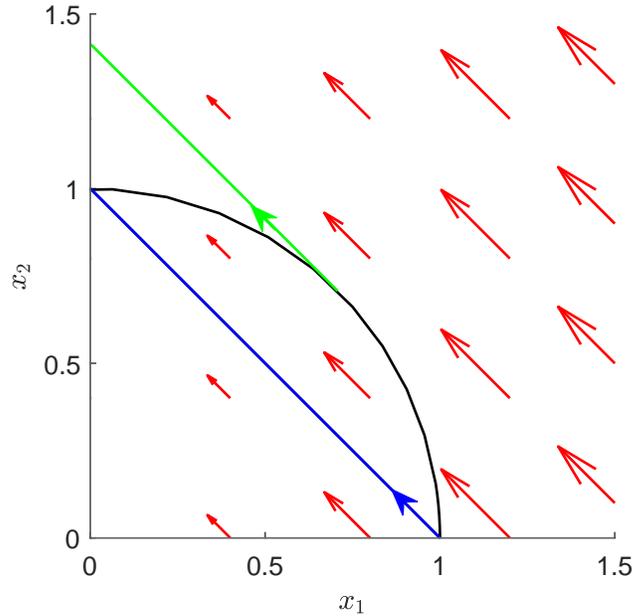}
\caption{The phase plane for system (\ref{eq:linear}) with 
$A=\begin{bmatrix}
-1&0\\
1&0
\end{bmatrix},$ which is reactive in $\ell_2$ but not in $\ell_1$. The line $x_1+x_2=1$ and the circle $x_1^2+x_2^2=1$ geometrically depict $||x||=1$ in the $\ell_1$ and $\ell_2$ norms respectively. The derivative vectors for the phase plane are shown in red and two different initial trajectories are shown in green and blue. The green trajectory is an example that is reactive in $\ell_2$, but not in $\ell_1$, and the blue trajectory is another example that is not reactive in $\ell_1$.}
\label{fig:reactl1}
\end{figure}

\subsection*{Example 1}
First, we present an example that is reactive in $\ell_2$ but not in $\ell_1$. Take system (\ref{eq:linear}) with 
\[A=\begin{bmatrix}
-1&0\\
1&0
\end{bmatrix}.\] This system simply redistributes individuals from patch 1 to patch 2, and the phase plane is shown in Figure  \ref{fig:reactl1}. It is not reactive in the $\ell_1$ norm because the total population size is not increasing, but it is reactive in $\ell_2$. This highlights how measuring reactivity in the $\ell_2$ norm can at times defy our biological expectation of what reactivity should mean --- the growth of a population --- and reinforces our rationale for using the $\ell_1$ norm to measure reactivity in metapopulations. While the matrix $A$ is reducible and this system is only semi-stable, and thus may be considered a borderline example, if $a_{22}$ is replaced by a small negative number, $-\epsilon$, and $a_{12}$ is replaced by a small positive number, $\epsilon/2$, then for sufficiently small $\epsilon$, $A$ will be irreducible and the system will now be stable, but will still be reactive in $\ell_2$ and not in $\ell_1$.

\subsection*{Example 2}
The second example, which is reactive in $\ell_1$ but not in $\ell_2$ is system (\ref{eq:linear}) with 
\[A=\begin{bmatrix}
-1&3/2\\
1/3&-1
\end{bmatrix},\] where the phase plane is shown in Figure \ref{fig:reactl2}. Now the system is reactive in $\ell_1$ because if we start with one individual on the second patch (the dynamics governed by the second row of $A$) the total population grows, but in such a way that it will not be reactive in $\ell_2$. This example demonstrates that again reactivity in $\ell_2$ can defy our biological expectation of reactivity, but now in the opposite way. Here the total population grows, yet the system is not reactive in $\ell_2$.  Note that this system is equivalent to system (\ref{eq:asymgrow1}) with $\epsilon=3$.

Together, the two examples highlight the differences that can occur when measuring reactivity in different norms and the caution that should be taken when interpreting reactivity in the $\ell_2$ norm biologically. Here we only present examples which are reactive in $\ell_2$ but not in $\ell_1$ and vice versa but it is also possible to find examples of systems which attenuate in $\ell_2$ but not in $\ell_1$.

\begin{figure}[!htb]
\centering
\includegraphics[width=9cm]{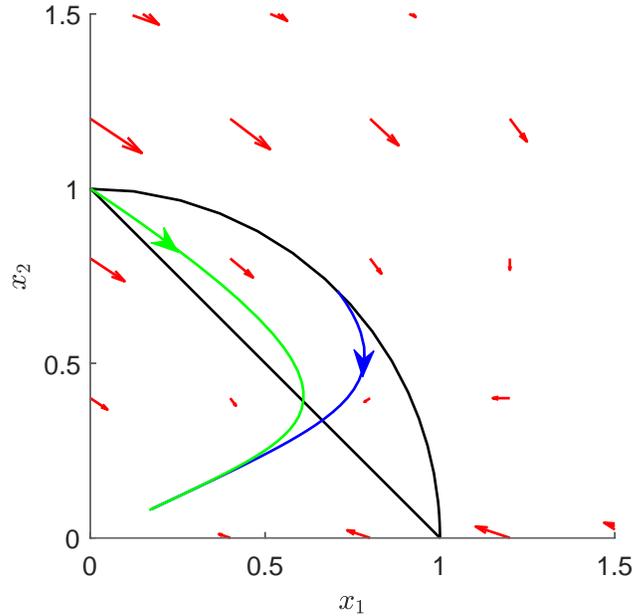}
\caption{The phase plane for system (\ref{eq:linear}) with 
$A=\begin{bmatrix}
-1&3/2\\
1/3&-1
\end{bmatrix},$ which is reactive in $\ell_1$ but not in $\ell_2$. The line $x_1+x_2=1$ and the circle $x_1^2+x_2^2=1$  geometrically depict $||x||=1$ in  the $\ell_1$ and $\ell_2$ norms respectively. The derivative vectors for the phase plane are shown in red and two different initial trajectories are shown in green and blue. The green trajectory is an example that is reactive in $\ell_1$, but not in $\ell_2$, and the blue trajectory is another example that is not reactive in $\ell_2$.}
\label{fig:reactl2}
\end{figure}

\subsection{The relationship between stability and reactivity/attenuation}

Now that we have presented a couple of examples that demonstrate the difference between reactivity in the $\ell_1$ and $\ell_2$ norms, we show that in both norms if a system is stable it attenuates, and if a system is unstable it is reactive. While we hypothesize that this is true in all norms, these are necessary properties for a norm to be useful in measuring reactivity or attenuation and thus we present the proofs in $\ell_1$ and $\ell_2$ for completeness.

We begin by presenting the proof for the $\ell_1$ norm. 

\begin{lemma}
If the $x=0$ equilibrium for \[x'=Ax\] is stable and $A$ is a Metzler matrix, then $\min \lambda_j<0$, where $\lambda_j=\sum_{i=1}^n a_{ij}$, and thus the system attenuates in the $\ell_1$ norm. Likewise if the $x=0$ equilibrium is unstable then $\max \lambda_j>0$ and so the system is reactive in the $\ell_1$ norm.
\end{lemma}

\begin{proof}

Let $A$ be a Metzler matrix, then \[\mu(A)=\max\{\Re (\lambda):\lambda \in \sigma(A)\},\] and $\mu(A)$ is an eigenvalue of $A$ with a non-negative normalized eigenvector $v$, with $\sum_{j=1}^n v_j=1$, such that $Av=\mu(A)v.$ (Thm A.43, \cite{Thieme2003}).

Therefore if the initial population is arranged according to the eigenvector $v$, then the initial growth rate is

\begin{align*}
\sum_{i=1}^n\frac{dx_i}{dt}&=1^Tx'|_{x_0=v}\\
&=1^TAv\\
&=1^T\mu(A)v.
\end{align*}

If the $x=0$ equilibrium for $x'=Ax$ is stable, then $\mu(A)<0$, and therefore \[\sum_{i=1}^n\frac{dx_i}{dt}=1^T\mu(A)v<0.\] Let $\lambda_k=\min_j \lambda_j$, then \[\lambda_k \leq \sum_{j=1}^n\lambda_jv_j = 1^TAv=1^T\mu(A)v<0,\] as $v$ is normalized. Thus if the $x=0$ equilibrium for $x'=Ax$ is stable, then $\min \lambda_j<0$, and so the system attenuates in the $\ell_1$ norm, and the minimum growth rate occurs when the initial population is all on a single patch.

Similarly if the $x=0$ equilibrium for $x'=Ax$ is unstable, then $\mu(A)>0$, and therefore \[\sum_{i=1}^n\frac{dx_i}{dt}=1^T\mu(A)v>0.\] Let $\lambda_k=\max_j \lambda_j$, then \[\lambda_k \geq \sum_{j=1}^n\lambda_jv_j = 1^TAv=1^T\mu(A)v>0.\] Thus if the $x=0$ equilibrium for $x'=Ax$ is unstable, then it is also reactive in the $\ell_1$ norm and the maximum growth rate occurs when the initial population is all on a single patch.

\end{proof}

In the $\ell_2$ norm the following lemma gives the corresponding result:

\begin{lemma}

If the $x=0$ equilibrium for $x'=Ax$ is stable then \[\min_{||x_0||_2\neq 0} \left(\frac{1}{||x||_2}\frac{d||x||_2}{dt}\right)\biggr\rvert_{t=0}<0\] and thus the system attenuates in $\ell_2$. Similarly if the $x=0$ equilibrium for $x'=Ax$ is unstable then 
\[\max_{||x_0||_2\neq 0} \left(\frac{1}{||x||_2}\frac{d||x||_2}{dt}\right)\biggr\rvert_{t=0}>0\] and so the system is reactive in $\ell_2$. 

\end{lemma}

\begin{proof}
In the $\ell_2$ norm 
\begin{equation}
\left(\frac{1}{||x||_2}\frac{d||x||_2}{dt}\right)\biggr\rvert_{t=0}=\frac{x_0^TH(A)x_0}{x_0^Tx_0},
\label{eq:rayleigh}
\end{equation} where $H(A)=(A+A^T)/2$ is the Hermitian part of $A$ \cite{Caswell1997}. Furthermore the maximum of the right hand side of equation (\ref{eq:rayleigh}) over all $x_0$ is equal to $\lambda_{\max}(H(A))$, where $\lambda_{\max}(H(A))$ is the maximum eigenvalue of $H(A)$ (Thm 4.2.2, \cite{Horn2012}). Similarly the minimum of the right hand side of equation (\ref{eq:rayleigh}) over all $x_0$ is equal to $\lambda_{\min}(H(A))$, the minimum eigenvalue of $H(A)$. Since $A$ is real, $H(A)$ is real and symmetric and so $\lambda_{\max}(H(A))$ and $\lambda_{\min}(H(A))$ are real.

Therefore $x'=Ax$ is reactive in the $\ell_2$ norm if $\lambda_{\max}(H(A))>0$ and attenuates if $\lambda_{\min}(H(A))<0$. To relate stability to reactivity and attenuation, we use the fact that the real part of the dominant eigenvalue of $A$ (denoted by $\mu(A)$) is between the maximum and minimum eigenvalues of the Hermitian part of $A$ (Fact 5, pg 14-2, \cite{Hogben2006}), \[\lambda_{\min}(H(A))\leq \mu(A) \leq \lambda_{\max}(H(A)).\] If the $x=0$ equilibrium for  $x'=Ax$ is unstable, then \[\lambda_{\max}(H(A))\geq \mu(A) \geq 0\] and so the system is reactive, and if the $x=0$ equilibrium for $x'=Ax$ is stable, then \[\lambda_{\min}(H(A))\leq \mu(A) \leq 0\] and so the system attenuates.

\end{proof}

In this section we have shown how to calculate reactivity and attenuation using the $\ell_1$ norm in metapopulations, proven that if the equilibrium of a system is unstable/stable then the system must be reactive/attenuate in both the $\ell_1$ and $\ell_2$ norms, and demonstrated the difference between reactivity in the $\ell_1$ and $\ell_2$ norms using a couple salient examples. We now return to the motivating feature of this paper --- systems that are reactive and stable or attenuate and are unstable --- and in the following section we provide examples of long lived transients in these systems.

\section{Metapopulations with arbitrarily large transient growth or decay}
\label{sec:arblarge}
Here we examine two different metapopulations, one of which is reactive and can exhibit arbitrarily large transient growth, and the other that attenuates and can decline to arbitrarily small levels. In each case this transient growth differs from the system's long term growth trajectory: the metapopulation that exhibits large growth eventually declines, and the system that declines eventually grows. Both of these example metapopulations are linear systems, and therefore the addition of non-linearities to construct more realistic models could further exacerbate the length of the transient period. These examples are not meant to imply that there are realistic biological metapopulations that can grow arbitrarily large before decaying, but rather to emphasize that the difference between transient dynamics and asymptotic dynamics can be quite stark even in linear systems.

\subsection{Arbitrarily large transient growth}
\label{sec:largegrowth}

First we present a reactive metapopulation that can exhibit arbitrarily large transient growth, but eventually declines. In this metapopulation individuals can either give birth to new individuals on the same patch, or give birth to individuals on the other patch, but there is no migration of individuals between patches. As mentioned in the introduction, this type of model is applicable to many marine metapopulations where adults are sedentary but larvae can disperse, to plant populations where seeds can be carried between habitat patches, or other populations governed by birth-jump processes. Let the metapopulation be described by:
\begin{align}
x'&=(b-d)x+b_{12}y\label{eq:arbgrow}\\
y'&=b_{21}x+(b-d)y\nonumber
\end{align}
so that $b$ is the on patch birth rate and $d$ is the on patch death rate, $b_{12}$ is the birth rate of individuals on patch 2 producing new individuals on patch 1, and $b_{21}$ is the birth rate of individuals on patch 1 producing new individuals on patch 2. The system is linear, so assuming that $(b-d)^2 \neq b_{12}b_{21}$ the only steady state is $x=y=0$.

For the metapopulation to eventually decline, both eigenvalues need to be negative. For system (\ref{eq:arbgrow}) the eigenvalues are 
\begin{align*}
\lambda_1&=(b-d)+\sqrt{b_{12}b_{21}}\\
\lambda_2&=(b-d)-\sqrt{b_{12}b_{21}}
\end{align*}
and thus we require that $(b-d)<0$ and $(b-d)^2>b_{12}b_{21}$. Now in order for the metapopulation to be reactive we need either $b_{12}>-(b-d)$ or $b_{21}>-(b-d)$. Here we choose $b_{21}>-(b-d)$, so that if we start with one individual on patch 1, i.e. $x(0)=1$, $y(0)=0$, the metapopulation initially grows. 

To prove that the metapopulation can grow arbitrarily large, we show that the limit as some parameter approaches 0 of $\max_t (x(t)+y(t))$ is unbounded. In terms of the maximum amplification defined by (\ref{eq:maxamp}), this is equivalent to showing that the limit of the maximum amplification in the $\ell_1$ norm, $\rho_{\max}$, becomes unbounded. This is because the initial condition chosen above is such that $x(t)+y(t)=\rho (t)$ in the $\ell_1$ norm . To take the limit, we must first reduce the parameters in our system until we are left with a single parameter which we can let approach 0, while still maintaining the inequalities above that govern the stability and reactivity of the system. Let $b-d=-1$, $b_{12}=\epsilon/2$, and $b_{21}=1/\epsilon$, where $\epsilon$ is a   small positive parameter that approaches 0.  Our reduced system can now be written as:
\begin{align}
x'&=-x+ \frac{{\epsilon}}{2} y\label{eq:asymgrow1}\\
y'&=\frac{1}{{{\epsilon}}} x-y\nonumber \\
x(0)&=1\nonumber\\
y(0)&=0\nonumber.
\end{align}
This system is stable and the digraph for this system is shown in  Figure \ref{fig:asymgrow}. This system is reactive in $\ell_1$ and $\ell_2$ for small $\epsilon$ and the solution is:
\begin{align*}
x(t)&=\frac{1}{2} \left(e^{-(1-\frac{1}{\sqrt{2}})t}+  e^{-(1+\frac{1}{\sqrt{2}})t}\right)\\
y(t)&=\frac{1}{\sqrt{2}\epsilon}\left(e^{-(1-\frac{1}{\sqrt{2}})t}-e^{-(1+\frac{1}{\sqrt{2}})t}\right).
\end{align*}

\begin{figure}
\centering
\begin{tikzpicture}
\node (x) [stage] {$x$};
\node (y) [stage, right of=x, xshift=1cm] {$y$};
\draw [arrow] (x) to [out=30,in=150] node[above] {$\epsilon^{-1}$} (y);  
\draw [arrow] (y) to [out=210,in=330] node[below] {$\epsilon/2$}(x);
\draw [arrow] (x) edge[out=210,in=150,loop] node[left] {$-1$} (x);
\draw [arrow] (y) edge[out=30,in =330,loop] node[right] {$-1$} (y);
\end{tikzpicture}
\caption{Digraph for system (\ref{eq:asymgrow1}). The directed edges represent the birth rate of individuals on the outgoing patch producing new individuals on the incoming patch. The self loops are the birth rate minus the death rate on a patch.}
\label{fig:asymgrow}
\end{figure}
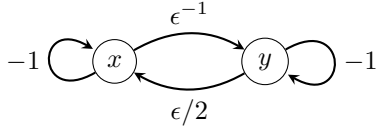

We want to show that 
\begin{equation*}
\lim_{\epsilon \rightarrow 0} \max_t \left(x(t)+y(t)\right)= \lim_{\epsilon \rightarrow 0} \rho_{\max} =\infty.
\end{equation*}
Normally to calculate the maximum we would take the derivative of $(x(t)+y(t))$, set it equal to 0, solve for $t$, and then evaluate $(x(t)+y(t))$ at this value of $t$. However it turns out this is rather complicated, so we will simplify this process by first noting that $x(t)>0$ for all $t$. Therefore 
\begin{equation*}
\max_t \left(x(t)+y(t)\right)> \max_t y(t).
\end{equation*}
Now we only have to perform the above process on $y(t)$, rather than $(x(t)+y(t))$. 

Setting $y'(t)=0$ and solving for $t$, we find that the time that the maximum of $y(t)$ is achieved, $t_{\rm max}$, along with the corresponding maximum in $y$, $y(t_{\rm max})$, are:
\begin{align*}
t_{\rm max}&= \frac{1}{\sqrt{2}}\left(\log(1+\sqrt{2})-\log(-1+\sqrt{2})\right)\\
y(t_{\rm max})&=\frac{\sqrt{2}}{\epsilon}(1+\sqrt{2})^{(-\frac{1}{2}-\frac{1}{\sqrt{2}})}(-1+\sqrt{2})^{(\frac{1}{\sqrt{2}}-\frac{1}{2})}.
\end{align*}
We can clearly see that $\lim_{\epsilon\rightarrow 0} y(t_{\rm max})=\infty$ and thus also  
$\lim_{\epsilon\rightarrow 0} \max_t x(t)+y(t)= \allowbreak  \lim_{\epsilon\rightarrow 0}\rho_{\max} \allowbreak =\infty$.

Therefore even in a two patch metapopulation that is asymptotically stable, there is always a parameter combination for which the total population, and thus also the maximum amplification in the $\ell_1$ norm, $\rho_{\max}$, can initially grow arbitrarily large before they decay. Again, this is  not meant to imply that there are realistic biological metapopulations that can grow arbitrarily large before decaying, but to emphasize how different the transient and asymptotic dynamics of a system can be.

\subsection{Transient decay to arbitrarily small levels}

Similar to the previous section, we now present an example of a metapopulation that attenuates and can decay to an arbitrarily small population size before eventually growing. This metapopulation is unstable mathematically because it eventually grows, but biologically the metapopulation could first go extinct if the total population size decays below one individual before it eventually increases. We  again use a metapopulation where individuals can either give birth to new individuals on their patch or on the other patch, but cannot migrate between patches. The difference between this metapopulation and the example used in the previous section, is that now the on patch birth and death rates differ between patches, but the between patch birth rates are the same. Let the metapopulation be described by:
\begin{align*}
x'&=(b_{1}-d_{1})x+\epsilon y\\
y'&=\epsilon x+(b_2-d_2)y
\end{align*}
where $b_1$ and $d_1$ are the birth and death rates on patch 1, $b_2$ and $d_2$ are the birth and death rates on patch 2, and $\epsilon$ is the interpatch birth rates for both patches.

\begin{figure}
\centering
\begin{tikzpicture}
\node (x) [stage] {$x$};
\node (y) [stage, right of=x, xshift=1cm] {$y$};
\draw [arrow] (x) to [out=30,in=150] node[above] {$\epsilon$} (y);  
\draw [arrow] (y) to [out=210,in=330] node[below] {$\epsilon$}(x);
\draw [arrow] (x) edge[out=210,in=150,loop] node[left] {$-1$} (x);
\draw [arrow] (y) edge[out=30,in =330,loop] node[right] {$\epsilon$} (y);
\end{tikzpicture}
\caption{Digraph for system (\ref{eq:asymp}). The directed edges represent the birth rate of individuals on the outgoing patch producing new individuals on the incoming patch. The self loops are the birth rate minus the death rate on a patch.}
\label{fig:asymdecay}
\end{figure}
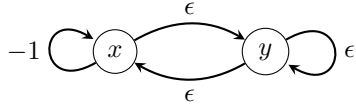

In order for the metapopulation to eventually grow, we assume that the birth rate is greater than the death rate on one of the patches. We choose this to be patch 2, thus we require $b_2-d_2>0$. We also want our population to initially decline when starting on patch 1, for this to occur we assume $b_1-d_1+\epsilon<0$. To prove that the metapopulation can decay to an arbitrarily small population size we reduce the system to have a single parameter and then show that the limit as the parameter approaches 0 of $\min_t (x(t)+y(t))=0$.  Let $b_1-d_1=-1$ and $b_2-d_2=\epsilon$, then our system can be written in terms of a single positive parameter, $\epsilon$, as:

\begin{align}
x'&=- x+\epsilon y \label{eq:asymp}\\
y'&=\epsilon x+ \epsilon y\nonumber\\
x(0)&=1\nonumber\\
y(0)&=0.\nonumber
\end{align}
This system is unstable and the corresponding digraph is shown in  Figure \ref{fig:asymdecay}. It attenuates in both the $\ell_1$ and $\ell_2$ norms for small $\epsilon$.

It is possible to show that the minimum population size can grow arbitrarily small in a manner similar to the previous section, though the calculations are somewhat more complicated. Instead in this section, we perform an asymptotic expansion in terms of $\epsilon$ to demonstrate the limiting behaviour of system (\ref{eq:asymp}). Let $x(t)=x_{0}(t)+\epsilon x_{1}(t)+O(\epsilon^2)$ and $y(t)=y_{0}(t)+\epsilon y_{1}(t)+O(\epsilon^2)$. Then the zero order system is:
\begin{align*}
x_{0}'(t)&=-x_{0}(t)\\
y_{0}'(t)&=0\\
x_{0}(0)&=1\\
y_{0}(0)&=0,
\end{align*}
that has the solution $x_{0}(t)=e^{-t}$ and $y_{0}=0$. We can proceed in a similar manner to solve the first order terms, and then our solution up to order $\epsilon$ is given by:
\begin{align*}
x(t)&=e^{-t}+O(\epsilon^2)\\
y(t)&=\epsilon(1-e^{-t})+O(\epsilon^2).
\end{align*}

This solution is valid for small $t$, and is therefore our inner approximation. To find our outer approximation for large $t$, we rescale $t=\tau/\epsilon$ and arrive at the system:
\begin{align*}
\epsilon X'&=-X+\epsilon Y\\
\epsilon Y'&=\epsilon X+\epsilon Y.
\end{align*}

We can again solve the zero order and first order equations and arrive at the following solution with two undetermined coefficients:
\begin{align*}
X(\tau)&=\epsilon C e^\tau+O(\epsilon^2)\\
Y(\tau)&=Ce^\tau+\epsilon(C\tau e^\tau+(C+K)e^\tau)+O(\epsilon^2).
\end{align*}
To solve our undetermined coefficients we require that $\lim_{t\rightarrow \infty} x(t)=\lim_{\tau \rightarrow 0} X(\tau)$, and $\lim_{t\rightarrow \infty} y(t)=\lim_{\tau \rightarrow 0} Y(\tau)$. From $x(\infty)=X(0^+)$, we find $C=0$. Substituting $C=0$ into $y(\infty)=Y(0^+)$ to solve for $K$ we find $K=1$. Adding our inner and outer solutions together and subtracting the overlap ($x(\infty)=X(0^+)=0$ and $y(\infty)=Y(0^+)=\epsilon$)  we find
\begin{align*}
x(t)&=e^{-t}+O(\epsilon^2)\\
y(t)&=\epsilon(e^{\epsilon t}-e^{-t})+O(\epsilon^2),
\end{align*}
thus our total population size behaves as 
\begin{equation}
x(t)+y(t)=e^{-t}+\epsilon(e^{\epsilon t}-e^{-t})+O(\epsilon^2).\label{eq:totalpop}
\end{equation}

We can see from equation (\ref{eq:totalpop}) and Figure \ref{fig:asympapprox} that for very small $\epsilon$, the total population size behaves similarly to $e^{-t}$ before eventually growing. Thus for an minimum population threshold, we can always find an $\epsilon$ small enough, such that the solution crosses the threshold before the population grows. We can prove this by solving the full system and taking the limit of the minimum, though we leave this out as the calculations become rather messy.

\begin{figure}
\centering
\includegraphics[width=9cm]{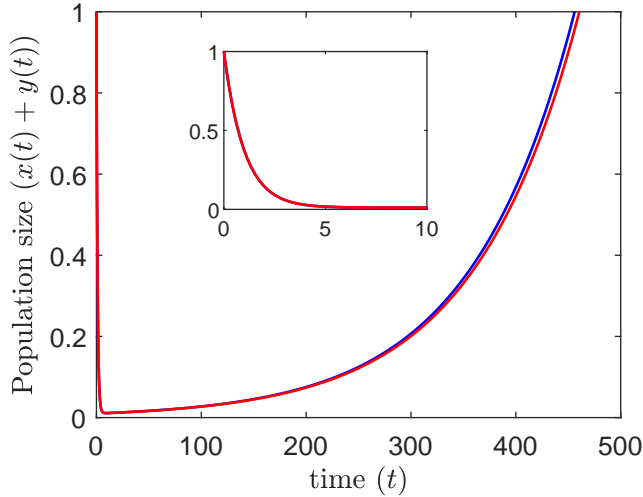}
\caption{Asymptotic approximation of the total population size (red) compared to the true total population size (blue) for system (\ref{eq:asymp}), with $\epsilon=0.01$. The asymptotic approximation is given by equation (\ref{eq:totalpop}).} 
\label{fig:asympapprox}
\end{figure}

Here we have shown that there are metapopulations for which the transient population can grow arbitrarily large or small, no matter the asymptotic stability of the system. In the next section we demonstrate how increasing the patch number can increase transient growth in advective metapopulations.

\section{Increasing patch number increases transient timescale}
\label{sec:largepatch}

In this section we examine the role that strong advection, coupled with a large number of patches, can have on the transient dynamics of metapopulations. In aquatic systems, habitat patches may be quite productive, but strong advection can sweep most larvae to the next patch, leading to large transient growth on downstream patches before the population eventually disappears from the last patch. This phenomenon can occur in metapopulations situated in rivers, ocean channels, or reef systems where reefs are arranged along a coastline with a directional current.


\begin{figure}
\centering
\begin{tikzpicture}
\node (x1) [stage2] {$x_1$};
\node (x2) [stage2, right of=x, xshift=1cm] {$x_2$};
\draw [arrow] (x1) to [out=30,in=150] node[above] {$b_2$} (x2);  
\draw [arrow] (x2) to [out=210,in=330] node[below] {$\epsilon$}(x1);
\draw [arrow] (x1) edge[loop below] node[below] {$b-d$} (x1);
\draw [arrow] (x2) edge[loop below] node[below] {$b-d$} (x2);
\node (xj-1) [draw=none,fill=none, right of=x2, xshift=0.5cm,minimum size=1cm] {$\dots$};
\node (xj) [stage2, right of=xj-1, xshift=0.5cm] {$x_j$};
\node (xj+1) [draw=none,fill=none, right of=xj, xshift=0.5cm,minimum size=1cm] {$\dots$};
\node (xn-1) [stage2, right of=xj+1, xshift=0.5cm] {$x_{n-1}$};
\node (xn) [stage2, right of=xn-1, xshift=0.5cm] {$x_n$};
\draw [arrow] (xn-1) to [out=30,in=150] node[above] {$b_2$} (xn);  
\draw [arrow] (xn) to [out=210,in=330] node[below] {$\epsilon$}(xn-1);
\draw [arrow] (xn-1) edge[loop below] node[below] {$b-d$} (xn-1);
\draw [arrow] (xn) edge[loop below] node[below] {$b-d$} (xn);
\draw [arrow] (xj) edge[loop below] node[below] {$b-d$} (xj);
\draw [arrow,dashed] (x2) to [out=30, in=150] node[above] {$b_2$} (xj-1);
\draw [arrow,dashed] (xj-1) to [out=210,in=330] node[below] {$\epsilon$}(x2);
\draw [arrow,dashed] (xj-1) to [out=30, in=150] node[above] {$b_2$} (xj);
\draw [arrow,dashed] (xj) to [out=210,in=330] node[below] {$\epsilon$}(xj-1);
\draw [arrow,dashed] (xj) to [out=30, in=150] node[above] {$b_2$} (xj+1);
\draw [arrow,dashed] (xj+1) to [out=210,in=330] node[below] {$\epsilon$}(xj);
\draw [arrow,dashed] (xj+1) to [out=30, in=150] node[above] {$b_2$} (xn-1);
\draw [arrow,dashed] (xn-1) to [out=210,in=330] node[below] {$\epsilon$}(xj+1);

\end{tikzpicture}
\caption{Digraph for system (\ref{eq:epspatch}). The directed edges represent the birth rate of individuals on the outgoing patch producing new individuals on the incoming patch. The self loops are the birth rate minus the death rate on a patch.}
\label{fig:longpatch}
\end{figure}
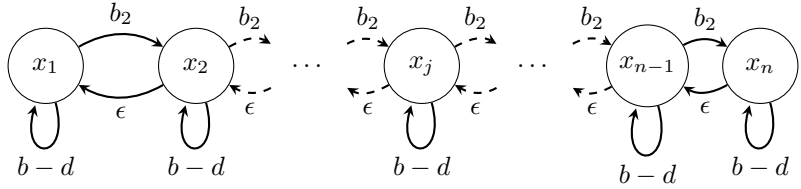

Consider a metapopulation on $n$ patches where the dynamics are described by the following system of equations:
\begin{align}
x'&=Ax \quad \label{eq:epspatch}\\
x(0)&=e_1\nonumber\\
A&=\begin{bmatrix}
b-d& \epsilon& 0 &\dots&0\\
b_2&b-d&\epsilon&\ddots&\vdots\\
0&b_2&b-d&\ddots&\\
\vdots& \ddots& \ddots&\ddots&\epsilon\\
0&\dots&0&b_2&b-d
\end{bmatrix},\nonumber
\end{align}
where $b$ is the on patch birth rate for each patch, $d$ is the death rate on each patch, $b_2$ is the birth rate of patch $j-1$ on patch $j$, $\epsilon$ is the birth rate from patch $j+1$ to patch $j$, and $e_1$ is a vector with 1 in the first entry and 0s elsewhere. The digraph for this system is shown in Figure \ref{fig:longpatch}.

The instantaneous measures of growth, $\lambda_j$, and the reactivity, $\bar{\sigma}$ in the $\ell_1$ norm, are therefore
\begin{align*}
\lambda_1&=b-d+b_2\\
\lambda_j&=b-d+b_2+\epsilon \quad j=2,\dots,n-1\\
\lambda_n&=b-d+\epsilon\\
\bar{\sigma}&=b-d+b_2+\epsilon.
\end{align*}
Let $b-d+\epsilon<0$ and $b-d+b_2>0$, then the system is reactive ($\bar{\sigma}>0$), and this maximum initial growth rate is achieved if the initial individual is on any patch except for patch 1 or $n$, though if the individual starts on patch 1 the initial growth rate is still positive. In system (\ref{eq:epspatch}) $A$ is a  tridiagonal Toeplitz matrix, so it has eigenvalues \cite{Noschese2013}
\begin{align*}
\lambda_h&=b-d+2\sqrt{b_2 \epsilon}\cos\left(\frac{h \pi}{n+1}\right)&h=1,\dots,n,
\end{align*}
and corresponding right eigenvectors, $v_h$, where the $k$th entry is given by 
\begin{align*}
v_{h,k}&=(b_2/\epsilon)^{k/2}\sin\left(\frac{hk\pi}{n+1}\right) &k=1,\dots,n;h=1,\dots,n.
\end{align*}

The solution to system (\ref{eq:epspatch}) can therefore be written as 
\begin{equation*}
x(t)=We^{Jt}W^{-1}e_1,
\end{equation*}
where $W$ is a matrix containing the eigenvectors, $v_h$, and $J$ is a diagonal matrix with the eigenvalues, $\lambda_h$, on the diagonal. For all but very small $t$, this is equivalent to the amplification envelope in the $\ell_1$ norm, $\rho (t)$, defined by equation  (\ref{eq:ampenv}).  Through examination of the eigenvalues, this system is stable if $\epsilon$ is small enough such that 
\begin{equation*}
b-d+2\sqrt{b_2\epsilon}<0.
\end{equation*}
Parameters that satisfy the inequalities that determine stability and reactivity in the $\ell_1$ norm can be found in the caption of  Figure \ref{fig:patches}. In this case the  maximum total population size, and also maximum amplification, are given by 
\begin{equation*}
x_{\max}=\rho_{\max}=\max_{t\geq 0}1^TWe^{Jt}W^{-1}e_1,
\end{equation*}
with the corresponding time $t_{\max}$, which is the value of $t$ for which the maximum occurs. The last measure of transience that is useful in this system is the total transient time, $t_{\rm total}$, which we define as the time it takes for the population size to decline below one, after initially starting with one individual, or

\begin{equation*}
t_{\rm total}=\min \{t>0:1^TWe^{Jt}W^{-1}e_1\leq 1\}.
\end{equation*}

So how does the number of patches affect the magnitude and length of transients? In Figure \ref{fig:patches}, which compares a 5 and 15 patch system, we can see that increasing the patch number increases both the magnitude of growth and the duration. 

\begin{figure}[!htb]
\centering
\subfloat[]{
\includegraphics[width=9cm]{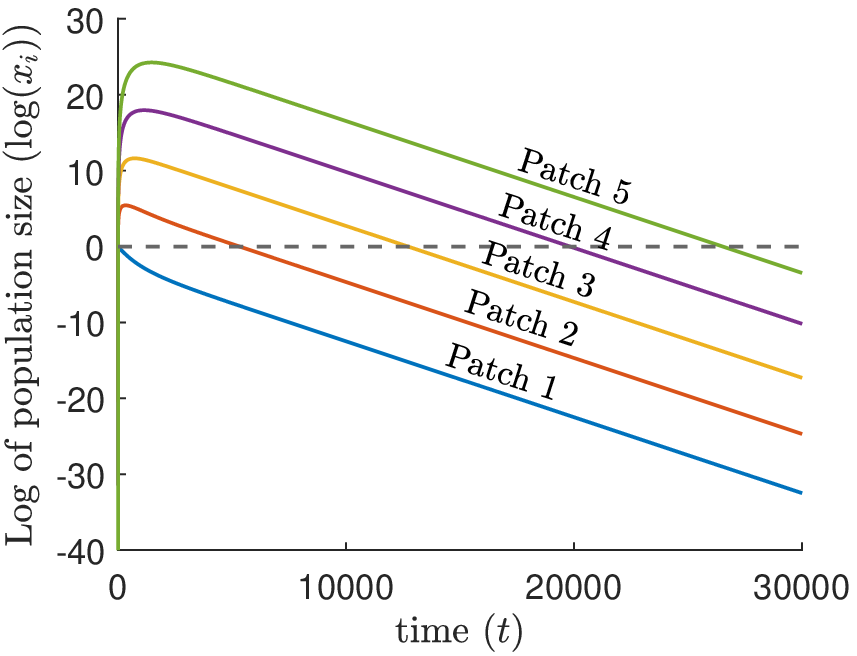}}

\subfloat[]{
\includegraphics[width=7.6cm]{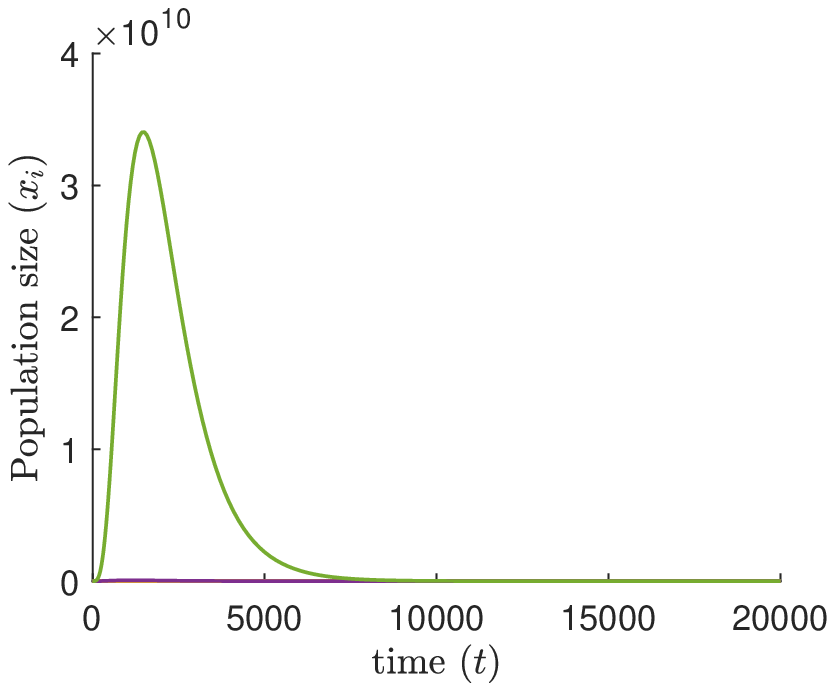}}
\subfloat[]{
\includegraphics[width=7.6cm]{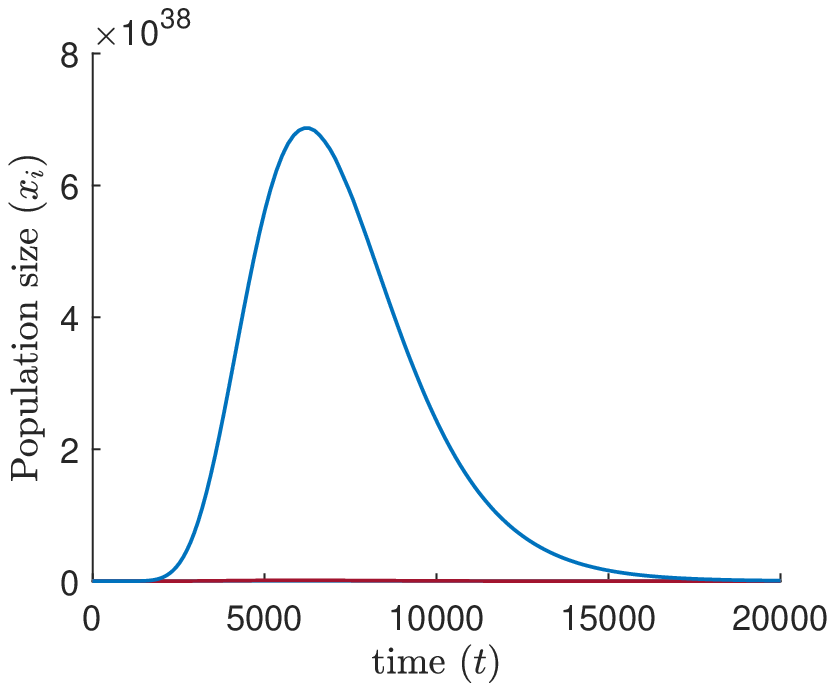}}

\caption{The population sizes on each patch for the advective system (\ref{eq:epspatch}), starting with one individual on patch 1. In a) the population sizes are shown on a log scale for a metapopulation of 5 patches, and in b) and c) the population sizes are shown for a metapopulation of 5 and 15 patches respectively. On the untransformed scale only the population size on the last patch can be seen as it is far larger than on any of the other patches, whereas on the log scale the population sizes of all patches can be seen. Parameters for this simulation are $b =0.09655$, $\epsilon =0.000001$, $b_2 =2$, and $d =0.1$, chosen so that system (\ref{eq:epspatch}) is reactive but stable.}
\label{fig:patches}
\end{figure}

Here it can be difficult to see the duration of transience exhibited by all patches on a regular scale, but on the log scale we can see that all patches except for patch 1 experience a large period of transient growth, before they decay below 1 individual (dashed line).  Patch 1 does not experience a large period of growth because the internal growth rate, $b-d$, is negative and the birth rate from patch 2 to patch 1, $\epsilon$, is too small to overcome this negative internal growth rate.

What cannot be seen from Figure \ref{fig:patches} is the dependence of transient growth on system parameters. We find that decreasing $b_2$ in system (\ref{eq:epspatch}) results in a large decrease in the maximum population size (and maximum amplification), $x_{\rm max}$ $(\rho_{max}$), and the total transient time, $t_{\rm total}$, but only a small decrease in the time at which the maximum population size is achieved, $t_{\max}$. Decreasing $b$ however, results in a large decrease in $x_{\rm max}$ ($\rho_{\max}$), $t_{\rm max}$, and $t_{\rm total}$. 

The relationship between increased transient time and number of patches can also be found for a linear metapopulation where all patches have negative initial growth rates, $\lambda_j$, except for the last patch which has a positive initial growth rate. In this case        the total population size decays for a long time before it eventually grows, and the time that it decays depends on the number of patches. 

We can see then that for a linear metapopulation, the length of the linear array can accentuate the transient growth that is possible in the system and that this is especially true for advective systems where there is some sort of directed birth in one direction in the array. Systems with this type of advective flow include  marine metapopulations located in channels near the mouth of rivers, or long coral reefs that are captured inside of a dominant coastal current flow. 

Having presented some illuminating metapopulation examples that demonstrate the magnitude that transient dynamics can differ from asymptotic dynamics, we now turn back to the general theory of transients in metapopulations and connect it to the source-sink distribution of habitat patches.

\section{Connecting the source-sink dynamics to the transient dynamics}

In this section we demonstrate how to connect the transient measures of initial population growth to the source-sink distribution of the metapopulation. For the transient measure of the patch specific contribution to the initial growth of the total population we use $\lambda_j$, previously defined in section \ref{sec:gentheory}. To measure the source-sink distribution we use the the classification derived from the next-generation matrix, $K$. 
This measure has previously been used to calculate the source-sink distribution in heterogenous environments \cite{Krkosek2010a, Mckenzie2012, Huang2016b, Harrington2020}. In order to calculate the next-generation matrix for system (\ref{eq:linear}) we decompose $A=F-V$, where $F$ is a non-negative matrix with positive entries that describe the birth of new individuals in the metapopulation, and $V$ is a non-singular $M$ matrix with entries that describe the transfer of individuals between compartments or in this case habitat patches, and also includes the death of individuals. Because $V$ is a non-singular $M$ matrix, $V^{-1}$ is non-negative. The next-generation matrix, $K=[k_{ij}]$, can then be calculated as $K=FV^{-1}$. This next generation matrix is then commonly used to calculate the basic reproduction number $R_0$, which is the average number of new individuals produced by one initial individual, as $R_0=\rho(K)$, where $\rho$ denotes the spectral radius. However we can also define $R_j$ as the number of new individuals produced on all patches from one initial individual starting on patch $j$, that can then be calculated as:
\[R_j=\sum_{i=1}^nk_{ij}.\] Using this measure we classify patch $j$ as a \textit{source} if $R_j>1$, as then one individual on patch $j$ would produce more than one individual in the total metapopulation. Likewise we classify patch $j$ as a \textit{sink} if $R_j<1$. 

\subsection{Expressing $R_0$ as a weighted sum of $R_j$}

Before examining the connection between the source-sink distribution $R_j$ and the initial growth $\lambda_j$, we first highlight a connection betwen $R_j$ and $R_0$. It turns out, as shown in the following Lemma, that $R_0$ can be calculated as a weighted sum of each $R_j$, and surprisingly this relationship between the spectral radius and the column sums of a matrix does not require any further assumptions on the matrix structure, though if the matrix is not non-negative, the components of the right eigenvector need not be real. Here $1$ is the (column) vector with each entry equal to 1, and $e_j$ is the vector with the only non-zero entry being 1 in the $j$th row.

\begin{lemma}
Let $v=[v_i]$ be the right eigenvector associated with the dominant eigenvalue of the next-generation matrix, $\mathcal{R}_0=\rho(K)$, normalized so $\sum_{1\leq i\leq n} v_i=1$. Then the basic reproduction number $\mathcal{R}_0=\sum_{1\leq j \leq n} R_jv_j$, where $R_j=1^T K e_j=\sum_{i=1}^n k_{ij}.$
\label{lemma:R0Ri}
\end{lemma}
\begin{proof}
First, we can rewrite $\mathcal{R}_0$ as
\begin{align*}
\mathcal{R}_0&=\mathcal{R}_0 1^T v\\
&=1^T \mathcal{R}_0 v,
\end{align*}
because the eigenvector has been normalized to sum to 1. Then, as $\mathcal{R}_0$ is the eigenvalue of $K$ associated with $v$ and the column sums of $K$ are $R_j$,

\begin{align*}
\mathcal{R}_0&=1^T \mathcal{R}_0 v\\
&=1^T Kv\\
&=\begin{bmatrix} R_1 & R_2 & \dots & R_n \end{bmatrix} v\\
&=\sum_{j=1}^n R_j v_j
\end{align*}
\end{proof}

The entries $v_j$ of the right eigenvector can be interpreted as the probability that a new individual begins on patch $j$ \cite{Cushing2016}. Therefore $\mathcal{R}_0$ can be interpreted as the sum over all patches, of the probability that an individual is born on patch $j$, multiplied by the number of new individuals it will produce on all other patches over its lifetime. 

Similarly, if we define $
\lambda_0$ to be the dominant eigenvalue of $A$, with the associated normalized eigenvector $u$, then 

\begin{align*}
\lambda_0&=1^TAu\\
&=\sum_{j=1}^n \sum_{i=1}^n a_{ij} u_j\\
&=\sum_{j=1}^n \lambda_j u_j,
\end{align*}
where it should be noted that $\lambda_j$ is the $j$th column sum of $A$, rather than an eigenvalue of $A$.

\subsection{Connecting the source-sink distribution, $R_j$, to the initial growth rate, $\lambda_j$}

Now that we have decomposed the dominant eigenvalues, $R_0$ and $\lambda_0$, into weighted sums of the columns of $K$ and $A$ respectively, we proceed to connect the source-sink distribution of a particular patch, $R_j$, to the initial growth from an individual on that patch, $\lambda_j$. To do so there are some restrictions that we need to impose on our metapopulation system and this is where we limit our study to marine or birth-jump metapopulation models where juveniles or seeds can disperse between patches while adults remain confined to habitat patches. The mathematical restriction defined by this class of models comes from the decomposition of $A$ into $F-V$. Here $V$ contains all entries that describe the transfer of individuals between compartments or patches. For the results presented in this paper, we require that $V$ has the following reducible form:
\begin{equation*}
V=\begin{bmatrix}
V_{11} & 0\\
V_{21} & D
\end{bmatrix},
\end{equation*}
where $V_{11}$ is $k\times k$, $D=\diag(d_{k+1},\dots,d_n)$ with $d_{k+1},\dots,d_n$ all positive, $0\leq k\leq n-1$, and $V$ is a non-singular $M$ matrix.

With this structure, individuals on patches $j=k+1,\dots,n$ cannot migrate between patches, but can still give birth to new individuals on any patch. If $V$ is completely diagonal, then there is no migration between any patches, only birth on other patches. This is the case for models of plants with seed dispersal, or simplified marine metapopulation models if the juvenile stage is not explicitly modelled. Under this structure, we first present a proofs connecting our instantaneous and generational growth measures, $\lambda_j$ and $R_j$, before presenting a two patch example.

\begin{lemma}
Let $A=F-V$ for system (\ref{eq:linear}), where $F$ is non-negative, and $V$ is a non-singular $M$ matrix with the following form: 
\begin{equation*}
V=\begin{bmatrix}
V_{11} & 0\\
V_{21} & D
\end{bmatrix},
\end{equation*}
where $V_{11}$ is $k\times k$, $D=\diag(d_{k+1},\dots,d_n)$ with $d_{k+1},\dots,d_n$ all positive, and $0\leq k\leq n-1$.  For $k+1\leq j\leq n$, $\lambda_j$ is positive if and only if $R_j>1$.
\label{thm:Riyi}
\end{lemma}
\begin{proof}
First, we can write $\lambda_j$ as
\begin{align*}
\lambda_j&=\sum_{i=1}^n a_{ij}\\
&=1^TAe_j.
\end{align*}

Then decomposing $A$ into $F-V$, and inserting $V^{-1}V$
\begin{align*}
\lambda_j&=1^T(F-V) e_j\\
&=1^T(F-V)V^{-1}V e_j\\ 
&=1^T(FV^{-1}-I)Ve_j.
\end{align*}

For $k+1\leq j \leq n$, $V$ is diagonal, so \[Ve_j=d_je_j.\] Therefore   
\begin{align*}
\lambda_j&=1^T(FV^{-1}-I)d_je_j\\
&=(R_j-1)d_j.
\end{align*} 
Now $d_j>0$, and thus $\lambda_j>0$ if and only if $R_j>1$.
\end{proof}

\begin{corollary}
\label{cor:diag}
In the notation of Lemma \ref{thm:Riyi}, if $V$ is diagonal, then $\lambda_j>0$ if and only if $R_j>1$ for $j=1,\dots,n$.
\end{corollary}

\begin{corollary}
\label{lem:sigma}
Under the same conditions as Lemma \ref{thm:Riyi}, in the $\ell_1$ norm $\bar{\sigma}>0$ if \\ $\max_{k+1\leq j\leq n}R_j>1$, and $\underline{\sigma}<0$ if $\min_{k+1\leq i\leq n}R_i<1$.
\end{corollary}

\begin{proof}
Under the conditions in Lemma \ref{thm:Riyi}, we know that $R_j-1$ has the same sign as $\lambda_j$ for $k+1\leq j\leq n$. Therefore if $\max_{k+1\leq j\leq n} R_j>1$ then $\bar{\sigma}=\max_{1\leq j\leq n} \lambda_j>0$, i.e. the system is reactive. Similarly if $\min_{k+1\leq j\leq n} R_j<1$ then $\underline{\sigma}=\min_{1\leq j\leq n} \lambda_j<0$, i.e. the population attenuates.
\end{proof}

\begin{corollary}
\label{cor:max}
Under the same conditions of Lemma \ref{thm:Riyi}, only with $V$ diagonal, then in the $\ell_1$ norm $\bar{\sigma}>0$ if and only if $\max_{1\leq j\leq n}R_j>1$ and $\underline{\sigma}<0$ if and only if $\min_{1\leq j\leq n}R_j$.
\end{corollary}
\begin{proof}
From Corollary \ref{cor:diag}, $\lambda_j>0$ if and only if $R_j>1$ for each patch $j$. Therefore if $\bar{\sigma}=\max_{1\leq j\leq n}\lambda_j>0$, then $\max_{1\leq j\leq n}R_j>1$, and likewise if $\max_{1\leq j\leq n}R_j>1$, then $\bar{\sigma}>0$. The same argument holds for $\min_{1\leq j\leq n} \lambda_j$ and $\min_{1\leq j\leq n}R_j$.
\end{proof}

Now that we have presented theory connecting the initial growth rate, $\lambda_j$, and the source-sink distribution, $R_j$, we present an example to illustrate how to calculate these measures and how  Lemma \ref{thm:Riyi} and the following corollaries can be used to connect them.

\subsection*{Example 3}\label{sec:2patch}

Here we present an example of a metapopulation consisting of two habitat patches, patch 1 and patch 2. New individuals can be born on either patch, but no individuals can migrate between patches. This system represents a simplification of the adult dynamics of many marine meroplanktonic metapopulations, where dispersal between patches occurs at the larval stage, rather than the sedentary adult stage. This system could also represent plant metapopulations that spread through seed dispersal, if the habitat landscape is patchy. The metapopulation dynamics can be represented with the following set of ODEs:
\begin{align}
x_1'&=b_{11}x_1+b_{12}x_2-d_1x_1\label{eq:patch1}\\
x_2'&=b_{21}x_1+b_{22}x_2-d_2x_2\nonumber,
\end{align}
where $b_{ij}$ is the birth rate for births from patch $j$ to patch $i$, and $d_i$ is the death rate on patch $i$. The lifecycle graph for this system is shown in Figure \ref{fig:patch1}. 

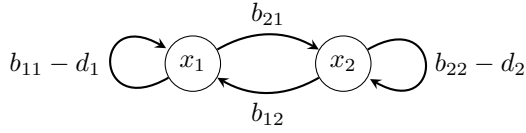
\begin{figure}
\centering
\begin{tikzpicture}
\node (x1) [stage] {$x_1$};
\node (x2) [stage, right of=x, xshift=1cm] {$x_2$};
\draw [arrow] (x1) to [out=30,in=150] node[above] {$b_{21}$} (x2);  
\draw [arrow] (x2) to [out=210,in=330] node[below] {$b_{12}$}(x1);
\draw [arrow] (x1) edge[out=210,in=150,loop] node[left] {$b_{11}-d_1$} (x1);
\draw [arrow] (x2) edge[out=30,in =330,loop] node[right] {$b_{22}-d_2$} (x2);
\end{tikzpicture}
\caption{Digraph for system (\ref{eq:patch1}). The directed edges represent the birth rate of individuals on the outgoing patch producing new individuals on the incoming patch. The self loops are the birth rate minus the death rate on a patch.}
\label{fig:patch1}
\end{figure}

We then decompose $A=F-V$ and construct the next generation matrix, $K=FV^{-1}$:
\begin{align*}
A&=\begin{bmatrix}
b_{11}-d_1 & b_{12}\\
b_{21} & b_{22}-d_2
\end{bmatrix}, \quad
F=\begin{bmatrix}
b_{11} & b_{12}\\
b_{21} & b_{22}
\end{bmatrix},\quad
V=\begin{bmatrix}
d_1 & 0\\
0 & d_2
\end{bmatrix},\\
K&=FV^{-1}=\begin{bmatrix}
b_{11}/d_1 & b_{12}/d_2\\
b_{21}/d_1 & b_{22}/d_2
\end{bmatrix}.
\end{align*}
For an initial individual starting on patch 1, the expected lifetime is $1/d_1$, and the rate that the individual is producing new individuals on both patches is $b_{11}+b_{21}$. Therefore $R_1=(b_{11}+b_{21})/d_1$ is the total number of individuals born onto both patch 1 and patch 2 over one generation. It is clear that $\lambda_1=b_{11}+b_{21}-d_1>0$ if and only if $R_1=\frac{b_{11}}{d_1}+\frac{b_{21}}{d_1}>1$, in accordance with Corollary \ref{cor:diag}.  Similarly $\lambda_2=b_{12}+b_{22}-d_2>0$ if and only if $R_2=\frac{b_{12}}{d_2}+\frac{b_{22}}{d_2}>1$. The system is therefore reactive if $\max(R_1, R_2)>1$ (Corollary \ref{cor:max}).

At first glance it seems obvious that if an individual starts on a source patch the population should have a positive initial growth rate or if the population starts on a sink patch it should have a negative initial growth rate, and we have shown from  Corollary \ref{cor:diag} and Example 3 
 that this is indeed the case for marine metapopulations. What is perhaps surprising is that this is not the case for general metapopulations when adults can migrate between habitat patches, and thus when the conditions of  Lemma \ref{thm:Riyi} and Corollary  \ref{cor:diag} are not met.  In the general case it is possible to start with an individual on a source patch, but for the population to initially decline and likewise to start on a sink patch but for the population to initially grow. An example of such a metapopulation is shown in Appendix \ref{appendix:a}.

Here in this section we have shown that for marine metapopulations and other metapopulations where the population dynamics are defined by birth jump processes, there is a one-to-one relationship between the source-sink distribution of patches and the initial growth rate when starting with one adult on a patch. That is, the initial population growth rate is positive if we start with one adult on patch $j$ if and only if patch $j$ is a source, and the initial growth rate is negative if and only if patch $j$ is a sink. This is a useful relationship biologically as there are several marine metapopulations where patches have already been classified into sources and sinks, and thus the transient dynamics for these systems can now be better understood. 

\section{Stage structure}

In this section we add stage structure to demonstrate some of the nuances in analysing transients in stage structured metapopulation models. The main issue with analysing reactivity and attenuation in models with stage structure is due to the fact that adults often give birth to many more juveniles than will survive to become adults, and that juveniles cannot normally give birth to new juveniles. This presents a few complications.

The first complication is the fact that if we want to analyse the initial growth or decay of a population, starting with an individual in a patch, it now depends if the individual is a juvenile or an adult. If we start with a juvenile, then there is no way that the total population, or even the patch population, can grow, given that the juvenile has to first survive to the adult stage to give birth to new juveniles. Thus we want to start with one adult on a patch.

However, if we start with an adult in a patch, and it gives birth to new juveniles, how do we count these new juveniles? If we are considering a marine metapopulation do we count every larvae as a new individual? If so, every marine metapopulation would exhibit transience, as each adult often produces thousands of larvae. This then begs the question: in a stage structured metapopulation, can we scale the juvenile population so that transient measures of population growth, such as reactivity and attenuation, are useful for stage structured models and measure the biologically relevant quantities?

To motivate the necessity of an honest scaling we highlight a discrete time example of transients in Dungeness crabs from Caswell and Neubert \cite{Caswell2005}. Dungeness crabs give birth to an enormous number of larvae, many of which do not survive to settle and become juveniles after 1 year. In this case the discrete time model requires a census time to measure new crabs after 1 year. If the census is taken pre breeding, then the system exhibits little reactivity, as many of the larvae that where initially born have not survived to become 1 year old juveniles. However, it the census is taken post breeding, then all of the eggs or larvae are counted and the initial amplification is increased by $10^5$. The models considered in this paper are continuous time and do not face this exact problem, but it is easy to see that the addition of a larvae stage in a marine stage structured metapopulation has large effects on the reactivity of the system.

Returning then to our stage structured model with only juvenile and adult stages, how should the juvenile stage be scaled so that an initial growth in juveniles also corresponds in some sense to growth in the total population? Ideally, we would scale the juvenile population so that each juvenile is scaled by the probability that it will become an adult. If we scale our population in this way then the measures of reactivity and attenuation regain their original meaning. If the maximum initial growth rate of our population, now scaled to be in terms of adults, is positive then our system is reactive, and if the minimum is negative then it attenuates. 

A biologically relevant measure of reactivity in a stage structured model must then be focused on the initial growth rate of the population, calculated so that the growth rates of juveniles are scaled by their contribution to the adult population. Under this scaling if any adult on any habitat patch produces many juveniles, but less than one become viable adults, then such a metapopulation is not reactive. Whereas if there is a patch such that one adult produces many juveniles and more than one survive to adulthood then the metapopulation is reactive, because the stage structured population, where juveniles are scaled according to their contribution to the adult population, is growing.

In the following sections we formally define such a scaling using a weighted $\ell_1$ norm, and contrast it with the unweighted $\ell_1$ norm that we have previously been using to calculated reactivity in metapopulations without stage structure.

\subsection{Unweighted $\ell^1$}

We first present the unweighted $\ell^1$ measure of the initial growth rate to demonstrate the mathematical framework that we use to examine reactivity in a stage-structured population with juveniles and adults. Now that there are both juveniles and adults we want to measure reactivity and attenuation as the total initial growth rate of the population, measured using either the weighted or unweighted norm, when we start with one adult on a patch. 

Consider a population on $n$ patches with a juvenile and adult stage. Let the population dynamics be described by 
\begin{equation}x'=Ax,\label{eq:patch} 
\end{equation} where $A$ is a $2n\times 2n$ matrix. Arrange $A$ so that all ODEs describing the change in the adult populations are in rows $n+1$ to $2n$. Assume that $A=F-V$, where $F$ is non-negative and $V$ is a non-singular M matrix so that $V^{-1}$ is non-negative. We are interested in metapopulations where adults can give birth to juveniles, but juveniles cannot give birth to new juveniles, so $F$ and $V$ can be written in block form as follows:
\begin{equation}
F=\begin{bmatrix}
0&F_{12}\\
0&0
\end{bmatrix}, \quad
V=\begin{bmatrix}
V_{11} & 0\\
V_{21} & V_{22}
\end{bmatrix}.\label{eq:vblock}
\end{equation}
With this decomposition $F_{12}$ contains all the new juvenile births from each adult patch, $V_{11}$ is a diagonal matrix that contains the rates of juvenile mortality on each patch as well as the maturation from juveniles to adults, $V_{22}$ is a diagonal matrix that contains the rates of adult mortality on each patch, and $V_{21}$ contains the negative of the rates of maturation/migration from juveniles to adults.

We define $\tilde{\lambda}_j$ to be the initial population growth rate, starting with one adult on patch $j$, measured using the $\ell_1$ norm. This can be defined mathematically for $1\leq j \leq n$ as 
\begin{equation*}
\tilde{\lambda}_j=\underbrace{\sum_{i=1}^nx_i'(0)}_\text{juvenile}+\underbrace{\sum_{i=n+1}^{2n}x_i'(0)}_\text{adult}, \quad x(0)=e_{j+n}.
\end{equation*}
In terms of $F$ and $V$ 
\begin{equation*}
\tilde{\lambda}_j=\sum_{i=1}^n f_{12ij}-\sum_{i=1}^n v_{22ij},
\end{equation*}
where $f_{12ij}$ and $v_{22ij}$ are the $(i,j)$ entries of $F_{12}$ and $V_{22}$ respectively. 

We use the tilde to differentiate the initial growth rate in the stage structured population, where we specifically begin with one adult on a patch, from the initial growth rate in a population without stage structure, where there is no difference in the type of individual that we start with. Having presented the mathematical framework that we use to measure reactivity in a stage structured population using the unweighted $\ell_1$ norm, we now use a weighted $\ell_1$ norm that  better captures the biological meaning of reactivity.

\subsection{Weighted $\ell^1$ for each patch}
\label{sec:weight}

In order to measure reactivity in a biologically meaningful fashion, we introduce a new measure of the initial population growth rate, $\tilde{\lambda}_j^p$. This initial population growth rate is calculated using a weighted $\ell_1$ norm so that the adult population is still measured using the regular $\ell_1$ norm, but the juvenile population on each patch is scaled by the probability that the juveniles survive to adulthood; the patch specific nature of the weighing is why we denote the initial growth rate $\tilde{\lambda}_j^p$. In this fashion $\tilde{\lambda}_j^p$ measures the initial growth rate of the total population if every member of the metapopulation was weighted according to how much they will contribute to the adult population. Adults are therefore not weighted, and juveniles are weighted by the probability that they survive to adulthood. This weighting recaptures the biological meaning of reactivity, where a system will only be reactive if the adult population will grow, and a system will not be reactive if there is only an initial spike in the juvenile population.

We use the same framework as the previous section to mathematically calculate $\tilde{\lambda}_j^p$, where we decompose $A=F-V$ and $F$ and $V$ are shown in block form in equation (\ref{eq:vblock}). Then we weight the juvenile population growth on each patch $i$ by a factor $s_i$, where $s_i$ is the probability that a juvenile from patch $i$ eventually becomes an adult. From the block form $V$, $s_i$ can be calculated as 
\begin{equation}
s_i=\sum_{k=1}^n(-V_{21}V_{11}^{-1})_{ki}.\label{eq:sj}
\end{equation}
To see how this corresponds to the probability of survival of a juvenile on patch $i$, consider the different block components of $F$ and $V$. The matrix $V_{11}^{-1}$ is diagonal, with the $(j,j)$ entry representing the average residence time of a juvenile born onto patch $j$. The matrix $-V_{21}$ contains the rates of maturation/migration of juveniles becoming adults on different patches, so the $(i,j)$ entry is the rate of maturation/migration of a juvenile on patch $j$ becoming an adult on patch $i$. This means that when we multiply $-V_{21}$ by $V_{11}^{-1}$ we are multiplying each of these rates by the residence times of the juveniles in the appropriate patches. In this way, the $(i,j)$ entry of  $-V_{21}V_{11}^{-1}$  is then  the probability that a juvenile leaving patch $j$ arrives on patch $i$. Therefore the $j$th column sum of $-V_{21}V_{11}^{-1}$ is the probability that a juvenile starting on patch $j$ becomes an adult on any other patch.

The initial growth rate using the weighted norm, $\tilde{\lambda}_j^p$, can then be calculated as the sum of the juvenile growth rates, each multiplied by the patch specific survival $s_i$, and the adult growth rates. Mathematically, this is defined as:
\begin{align*}\tilde{\lambda}_j^p&=\underbrace{\sum_{i=1}^ns_ix_i'(0)}_\text{juvenile}+\underbrace{\sum_{i=n+1}^{2n}x_i'(0)}_\text{adult}, \quad x(0)=e_{j+n}\\
&=\sum_{i=1}^ns_if_{12ij}-\sum_{i=1}^nv_{22ij}.
\end{align*}

In order to demonstrate that the initial growth rate calculated using the weighted $\ell_1$ norm, $\tilde{\lambda}_j^p$, indeed measures the growth rate of the population if all individuals are weighted according to their contribution to the adult population, we show that $\tilde{\lambda}_j^p$ is equivalent to scaling the juvenile population on each patch by the probability of survival to adulthood, and then measuring the initial growth rate using the unweighted $\ell_1$ norm, defined previously as $\tilde{\lambda}_j$.

\begin{lemma}
If each juvenile population on patch $j$ in system (\ref{eq:patch}) is rescaled by the patch specific survival probability, $s_i=\sum_{k=1}^n(-V_{21}V_{11}^{-1})_{ki}$, then the initial growth rate using the unweighted $\ell^1$ norm, $\tilde{\lambda}_j$, is equal to the patch specific weighted initial growth rate, $\tilde{\lambda}_j^p$ for the unscaled system. 
\end{lemma}

\begin{proof}
Rescale the juvenile population on patch $i$ by the patch specific survival probability $s_i$ given in equation (\ref{eq:sj}). In terms of system (\ref{eq:patch}) this means that $x_i^*=s_ix_i$ for $i=1,\dots,n$, $x_i^*=x_i$ for $i=n+1,\dots,2n$. Rewriting the system of equations 
\begin{align*}
x^{*\prime}&=A^*x^*\\
A^*=F^*-V^*,\quad
F^*&=\begin{bmatrix}
0 &SF_{12}\\
0 &0
\end{bmatrix},\quad
V^*=\begin{bmatrix}
V_{11} & 0 \\
V_{21}S^{-1} & V_{22}
\end{bmatrix},
\end{align*}
where $S=\diag(s_1,\dots,s_n)$. The unweighted initial growth rate for the scaled system is then 
\begin{align*}
\tilde{\lambda}_j^*&=\sum_{i=1}^n(SF_{12})_{ij}-\sum_{i=1}^nv_{22ij}\\
&=\sum_{i=1}^n s_if_{12ij}-\sum_{i=1}^nv_{22ij}\\
&=\tilde{\lambda}_j^p.
\end{align*}
Thus the unweighted initial growth rate for the scaled system is equal to the patch-weighted initial growth rate of the unscaled system, $\tilde{\lambda}_j^p$.

\end{proof}

We believe that it is more intuitive to measure reactivity in a stage-structured system using a weighted norm, rather than scaling the juvenile population and  using the unweighted $\ell_1$ norm, but for other systems this may not be the case. Recently Mari et al. \cite{Mari2017} have developed a new measure of reactivity called generalized reactivity, or g-reactivity, so that the reactivity of any specific combination of state variables in a system can be measured, and we demonstrate how to place our work in this context. The general framework of g-reactivity allows the reactivity of only the predator to be measured in a predator-prey system, or a single patch in a metapopulation model. Moreover in a stage-structured model, g-reactivity can be used to allow for a differential contribution of the juvenile and adult populations to the reactivity of the system, and so we can compare the calculation of g-reactivity to our calculation using the weighted $\ell_1$ norm. To calculate the g-reactivity of a system $x'=Ax$, a linear transformation is introduced, $y=Cx$, where $C$ is a matrix that defines the required contribution of each of the state variables, and then reactivity is calculated for $y$ using equation (\ref{eq:reactive}). For system (\ref{eq:patch}), g-reactivity is equivalent to reactivity in our weighted norm, $\max_j \tilde{\lambda}_j^p$, if $C$ is a $2n\times 2n$ identity matrix, but with the first $n$ diagonal entries replaced with $s_1,\dots,s_n$.

Returning to our measure of initial growth rate using a patch weighted norm, we present a couple of examples below to illustrate the calculation of $\tilde{\lambda}_j^p$ in different systems.

\subsection*{Example 4}
Consider a two patch system where juveniles are born onto all patches but only mature into adults on the patch where they were born:
\begin{align}
j_1'&=b_{11}a_1+b_{12}a_2-m_1j_1-d_{j1}j_1 \label{eq:ja1}\\
j_2'&=b_{21}a_1+b_{22}a_2-m_2j_2-d_{j2}j_2 \nonumber\\
a_1'&=m_1j_1-d_{a}a_1\nonumber\\
a_2'&=m_2j_2-d_{a}a_2.\nonumber
\end{align}
Here $j_i$ is the number of juveniles on patch $i$, $a_i$ is the number of adults on patch $i$, $b_{ij}$ is the birth rate of juveniles on patch $i$ from adults on patch $j$, $m_i$ is the maturation rate of juveniles on patch $i$ into adults on patch $i$, $d_{ji}$ is the death rate of juveniles on patch $i$, and $d_a$ is the death rate of adults, which is the same on both patches. The lifecycle graph for this system is shown in Figure (\ref{fig:ja1}).

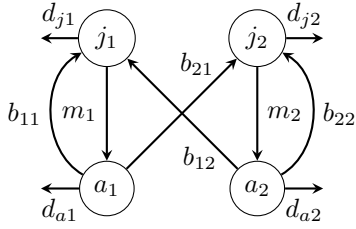
\begin{figure}
\centering
\begin{tikzpicture}
\node (j1) [stage] {$j_1$};
\node (j2) [stage, right of=j1, xshift=1cm] {$j_2$};
\node (a1) [stage, below of=j1, yshift=-1cm] {$a_1$};
\node (a2) [stage, below of=j2, yshift=-1cm] {$a_2$};
\draw [arrow] (a1) to [out=150,in=210] node[left] {$b_{11}$} (j1);
\draw [arrow] (j1) to [out=270, in=90] node[left] {$m_1$} (a1); 
\draw [arrow] (a1) edge  node[left,pos=0.95, xshift=-0.1cm] {$b_{21}$} (j2); 
\draw [arrow] (a2) edge  node[left,pos=0.1] {$b_{12}$} (j1); 
\draw [arrow] (j2) edge  node[right] {$m_{2}$} (a2); 
\draw [arrow] (a2) to [out=30,in=330] node[right] {$b_{22}$} (j2);
\node (j0) [draw=none,fill=none, left of=j1] {};
\node (a0) [draw=none,fill=none, left of=a1] {};
\node (j3) [draw=none,fill=none, right of=j2] {};
\node (a3) [draw=none,fill=none, right of=a2] {};
\draw [arrow] (j1) edge node[above] {$d_{j1}$} (j0);
\draw [arrow] (a1) edge node[below] {$d_{a1}$} (a0);
\draw [arrow] (j2) edge node[above] {$d_{j2}$} (j3);
\draw [arrow] (a2) edge node[below] {$d_{a2}$} (a3);
\end{tikzpicture}
\caption{Digraph for system (\ref{eq:ja1}). Here $b_{ij}$ is the birth rate of juveniles on patch $i$ from adults on patch $j$, $m_i$ is the maturation rate of juveniles on patch $i$ to adults on patch $i$, $d_{ji}$ is the death rate of juveniles on patch $i$ and $d_{ai}$ is the death rate of adults on patch $i$.}
\label{fig:ja1}
\end{figure}

In this case

\begin{align*}
F_{12}=\begin{bmatrix}
b_{11}&b_{12}\\
b_{21}&b_{22}
\end{bmatrix},\quad
V_{11}&=\begin{bmatrix}
m_1+d_{j1}&0\\
0&m_2+d_{j2}
\end{bmatrix},\quad
V_{21}=\begin{bmatrix}
-m_1&0\\
0&-m_2
\end{bmatrix},\\
V_{22}=\begin{bmatrix}
d_a & 0 \\ 0 & d_a
\end{bmatrix}, \quad
-V_{21}V_{11}^{-1}&=\begin{bmatrix}
\frac{m_1}{m_1+d_{j1}}&0\\
0&\frac{m_2}{m_2+d_{j2}}
\end{bmatrix},\\
\tilde{\lambda}_1^p&=s_1b_{11}+s_2b_{21}-d_{a}\\
\tilde{\lambda}_2^p&=s_1b_{12}+s_2b_{22}-d_{a}\\
s_1&=\frac{m_1}{m_1+d_{j1}}\\
s_2&=\frac{m_2}{m_2+d_{j2}}.
\end{align*}

If we look at $\tilde{\lambda}_1^p$, we see that $s_1$ is the probability that a juvenile born onto patch 1 survives to become an adult and it is multiplying $b_{11}$, the birth rate of juveniles onto patch 1 from adults on patch 1. Therefore the first component of $\tilde{\lambda}_1^p$ represents the rate of birth of new juveniles onto patch 1 from one adult on patch 1, but scaled by the probability that these juveniles survive to become adults. Likewise the second component, $s_2b_{21}$, is the rate of birth of new juveniles onto patch 2 from one adult on patch 1, scaled by the probability that those juveniles also become adults. Thus $\tilde{\lambda}_1^p$ is the initial growth rate of the total population, scaled in terms of the contribution to the adult population, when the population begins with one adult on patch 1.

\subsection*{Example 5}

We also consider a system in which juveniles are born onto the same patch as adults, but can then migrate between patches as they mature into adults:

\begin{align}
j_1'&=b_{11}a_1-m_{11}j_1-m_{21}j_1-d_{j1}j_1\label{eq:ja2}\\
j_2'&=b_{22}a_2-m_{22}j_2-m_{12}j_2-d_{j2}j_2\nonumber\\
a_1'&=m_{11}j_1+m_{12}j_2-d_{a}a_1\nonumber\\
a_2'&=m_{22}j_2+m_{21}j_1-d_{a}a_2.\nonumber
\end{align}
from which we calculate
\begin{align*}
F_{12}=\begin{bmatrix}
b_{11} & 0\\ 0 &b_{22}
\end{bmatrix}, \quad
&V_{11}=\begin{bmatrix}
m_{11}+m_{21}+d_{j1} & 0 \\ 0 & m_{12}+m_{22}+d_{j2}
\end{bmatrix}, \quad
V_{21}=\begin{bmatrix}
-m_{11} & -m_{12}\\ -m_{21} &-m_{22}
\end{bmatrix},\\
V_{22}=\begin{bmatrix}
d_a & 0 \\ 0 & d_a
\end{bmatrix}, \quad
-&V_{21}V_{11}^{-1}=\begin{bmatrix}
\frac{m_{11}}{m_{11}+m_{21}+d_{j1}} & \frac{m_{12}}{m_{12}+m_{22}+d_{j2}}\\
\frac{m_{21}}{m_{11}+m_{21}+d_{j1}} & \frac{m_{22}}{m_{12}+m_{22}+d_{j2}}
\end{bmatrix},\\
\tilde{\lambda}_1^p&=s_1b_{11}-d_a\\  
\tilde{\lambda}_2^p&=s_2b_{22}-d_a\\  
s_1&=\frac{m_{11}+m_{21}}{m_{11}+m_{21}+d_{j1}}\\
s_2&=\frac{m_{12}+m_{22}}{m_{12}+m_{22}+d_{j2}}.
\end{align*}

\begin{figure}
\centering
\begin{tikzpicture}
\node (j1) [stage] {$j_1$};
\node (j2) [stage, right of=j1, xshift=1cm] {$j_2$};
\node (a1) [stage, below of=j1, yshift=-1cm] {$a_1$};
\node (a2) [stage, below of=j2, yshift=-1cm] {$a_2$};
\draw [arrow] (a1) to [out=150,in=210] node[left] {$b_{11}$} (j1);
\draw [arrow] (j1) to [out=270, in=90] node[left, xshift=0.1cm] {$m_{11}$} (a1); 
\draw [arrow] (j2) edge  node[left,pos=0.05, xshift=-0.1cm] {$m_{12}$} (a1); 
\draw [arrow] (j1) edge  node[left,pos=0.9] {$m_{21}$} (a2); 
\draw [arrow] (j2) edge  node[right,xshift=-0.1cm] {$m_{22}$} (a2); 
\draw [arrow] (a2) to [out=30,in=330] node[right] {$b_{22}$} (j2);
\node (j0) [draw=none,fill=none, left of=j1] {};
\node (a0) [draw=none,fill=none, left of=a1] {};
\node (j3) [draw=none,fill=none, right of=j2] {};
\node (a3) [draw=none,fill=none, right of=a2] {};
\draw [arrow] (j1) edge node[above] {$d_{j1}$} (j0);
\draw [arrow] (a1) edge node[below] {$d_{a1}$} (a0);
\draw [arrow] (j2) edge node[above] {$d_{j2}$} (j3);
\draw [arrow] (a2) edge node[below] {$d_{a2}$} (a3);
\end{tikzpicture}
\caption{Digraph for system (\ref{eq:ja2}). Here $b_{ii}$ is the birth rate of juveniles on patch $i$ from adults on patch $i$, $m_{ij}$ is the maturation rate of juveniles on patch $j$ to adults on patch $i$, $d_{ji}$ is the death rate of juveniles on patch $i$ and $d_{ai}$ is the death rate of adults on patch $i$.}
\label{fig:ja2}
\end{figure}
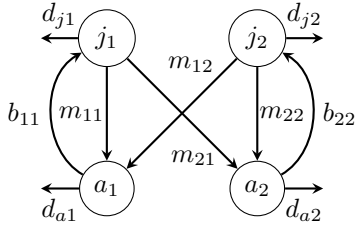

Now if we examine the first component of $\tilde{\lambda}_1^p$, $s_1=(m_{11}+m_{21})/(m_{11}+m_{21}+d_{j1})$, we can see that because juveniles from patch 1 can now migrate (as they mature) to both patches, $s_1$ is the probability that juveniles from patch 1 become adults on either patch. Likewise $s_2$ is the probability that juveniles from patch 2 become adults on either patch.

Here we have shown that if we use a weighted $\ell_1$ norm to scale the initial growth rate so that the juvenile population is scaled by the patch specific probability that juveniles become adults, then our scaled initial growth rate, $\tilde{\lambda}_j^p$ matches the biological intuition that we would like when measuring initial growth of the population. It is positive if the population, scaled so that every individual is measured by its contribution to the adult population, is growing, and negative if the population is decreasing. Measures of reactivity and attenuation then also represent their intuitive biological properties, and we are no longer in the situation (as if the initial growth rate was unscaled) that most marine metapopulations are reactive.

It is also possible to create a weighted norm where the juvenile populations on each patch are weighted by the same probability of survival, rather than by patch-specific probabilities $s_i$. In some cases it may be useful to scale all patches by the same survival probability, though under this weighted norm reactivity no longer corresponds exactly to the intuitive biological meaning mentioned previously.

\section{Discussion}

Transient dynamics often differ drastically from the asymptotic dynamics of a system and require different analytical tools. In this paper we have presented a unifying framework for analysing transient dynamics in marine and other birth-jump metapopulations, from the choice of norms to the incorporation of stage structure. We began by using the $\ell_1$ norm to define reactivity and attenuation in single species metapopulations and used examples to compare reactivity in the $\ell_1$ norm with reactivity in the more commonly used $\ell_2$ norm. We presented a couple of two patch models that gave rise to long transients: one stable system that exhibits a long period of growth before eventual decay, and one unstable system that exhibits a long period of decay before growth. In marine metapopulations, where patches are connected via larval dispersal, we examined how strong advective flow, coupled with a large number of patches, can lead to large transient growth. We then connected the initial growth rate of the metapopulation to the source-sink distribution of patches and lastly we demonstrated how to measure reactivity meaningfully in stage-structured marine metapopulations.

We are by no means the first to analyze the transient dynamics of systems, and in fact there has been an increase in the study of transient dynamics over the last couple of decades. In a pair of recent papers, several authors  have identified certain mechanisms as the main causes of long transients in ecological systems \cite{Hastings2018, Morozov2020}. These identified mechanisms that cause the long transients present in the examples in this paper are slow-fast systems, crawl-bys, and high dimensionality. Slow-fast systems cause long transients when the system rapidly converges to a slow manifold, then moves slowly towards or away from an equilibrium, depending on the stability of the system. This occurs in both examples in section \ref{sec:arblarge}. The second example in section   \ref{sec:arblarge} is also an instance of a crawl-by where the initial condition is near a saddle equilibrium but the movement away from the equilibrium occurs over a long timescale. Lastly in section  \ref{sec:largepatch} we explicitly demonstrated how increasing the dimension of a system, by increasing the patch number in a linear metapopulation,  leads to longer transients.

Our work also reinforces the fact that reactivity is a property specific to the norm under which it is measured. This has been mentioned in the first paper on reactivity by Neubert and Caswell\cite{Caswell1997}, who also recognize that it is always possible to find a norm such that a stable system is never reactive. It has also been noted by Lutscher and Wang \cite{Lutscher2020}, who mention that reactivity must be analysed in the dimensional version of a system rather than the non-dimensionalized version. The reactivity may be different between the two systems but the dimensional system is where the measure of reactivity is biologically meaningful. When analysing reactivity in metapopulations this fact is significant in two ways: first by using the $\ell_1$ norm rather than the $\ell_2$ norm to measure reactivity we can explicitly measure the growth rate of a population, and second by using a weighted $\ell_1$ norm we prevent the juvenile population from disproportionately affecting the reactivity of the system.

Differentially weighting certain classes of a population to calculate reactivity has been mentioned in passing by Verdy and Caswell \cite{Verdy2008}, and more extensively by Mari et al. \cite{Mari2017} who developed a new measure of reactivity called general reactivity, or g-reactivity. This is a method of only measuring the reactivity of the components of interest in a population, e.g. predators in a predator-prey model, and can also be used more generally to scale the contribution of different components of the population. Our method of using a weighted $\ell_1$ norm for stage-structured models has an equivalent formulation using g-reactivity that is discussed in  section \ref{sec:weight}, though Mari et al. \cite{Mari2017} use the $\ell_2$ norm to measure reactivity, rather than the $\ell_1$ norm, and are thus using a different measure of population growth.

While we believe the $\ell_1$ norm is the most biologically relevant norm to measure reactivity, we are among the first to use it to analyse reactivity in continuous time models. Townley et al. \cite{Townley2007} show how to calculate reactivity for stage-structured models in continuous time using the $\ell_1$ norm, but in following papers proceed to analyse reactivity in the $\ell_1$ norm only in discrete time systems \cite{Townley2008, Stott2010, Stott2011, Stott2011a}. Most authors measure reactivity with the $\ell_2$ norm, presumably due to the nice mathematical property that reactivity in the $\ell_2$ norm can be measured simply as $\lambda(H(A))$, where $H(A)=(A+A^T)/2$ \cite{Caswell1997, Caswell2005, Neubert2002, Neubert2004, Verdy2008, Lutscher2020, Snyder2010}. But while mathematically tractable, the biological meaning of Euclidean distance ($\ell_2$) is less clear than population size ($\ell_1$) and as shown in section \ref{sec:reacl1l2}, there are times when reactivity in $\ell_2$ does not correspond to an increase in population size.

In addition to being biologically meaningful, reactivity in the $\ell_1$ norm is also simple to calculate for single species metapopulations due to the fact that the populations on each patch are non-negative, and thus the full solution remains in the non-negative cone. Mathematically, a metapopulation governed by $x'=Ax$ will be reactive in the $\ell_1$ norm if $\max_j \sum_{i=1}^n a_{ij}>0$. Geometrically, a system is reactive in the $\ell_1$ norm if the dot product of the derivative vector of any initial condition and the outward normal vector of the plane $x_1+x_2+\dots x_n=1$ is positive.  If we instead want to consider the reactivity in a system where this assumption is relaxed, say a predator-prey metapopulation, then the mathematical calculation using the $\ell_1$ norm is no longer so simple, but our geometric intuition is similar. Instead of verifying if the dot product of the derivative vector of any initial condition and the outward normal vector of the plane $x_1+x_2+\dots x_n=1$ is positive, see for example Figure \ref{fig:reactl1}, we would instead need to take the dot product with the outward normal vector to the hypercube $|x_1|+|x_2|+\dots |x_n|=1$. If we continue to use $\max_j \sum_{i=1}^n a_{ij}>0$ to calculate reactivity for systems that do not remain in the non-negative cone then it is possible for systems to not be reactive and yet also be unstable, thus an interesting area for future work would be to mathematically formulate reactivity in terms of the matrix $A$ for these metapopulations.

No matter the norm in which reactivity and attenuation are measured, they are defined in terms of the linearization of a non-linear system around an equilibrium. As mentioned in the Introduction,  
reactivity and attenuation are therefore most relevant around hyperbolic equilibria, where the dynamics of the non-linear system are well approximated by the linearized system. Another caveat is that it is possible for an equilibrium of a non-linear system to not be reactive,  but for a perturbation of the non-linear system to still cause a large excursion away from the stable equilibrium before eventually returning. Excitable systems, such as the FitzHugh-Nagumo system, have stable equilibria with attracting regions, but small perturbations still trigger large excitations \cite{FitzHugh1961, Nagumo1962}. These systems may not be reactive from the linearized definition of reactivity, but can still exhibit similar behaviour to reactivity in the non-linear system, given a sufficient perturbation.

The final extension that we would like to highlight is the relationship between reactivity of continuous time models and reactivity of their discrete counterparts. Many marine metapopulations are modelled in discrete time due to yearly breeding cycles, however some are modelled in discrete time due to ease of simulation. For these models, where the time step is on the order of hours or days, we can connect the reactivity of the continuous time system with the discrete time system using a Taylor expansion. The continuous time system, $x'=Ax$, has the solution $x(t)=e^{At}x_0$ that could be sampled at discrete time steps $\tau$ to create the discrete time system $x(t+\tau)=B x(t),$ where $B=e^{A \tau}$.

The continuous time system $x'=Ax$ is reactive in $\ell_1$ if $A$ has a positive column sum (Lemma \ref{lemma:reactive}). In discrete time the system $x(t+\tau)=Bx(t)$ is reactive in $\ell_1$ if $B$ has a column sum that is greater than 1 \cite{Townley2007}. Assuming $\tau$ is a small timestep then we can approximate $B=e^{A \tau}=I+A \tau +O(\tau^2)$. Thus we can see that if the system is reactive in continuous time, i.e. there is a positive column sum of $A$, then we can find a sufficiently short time step $\tau$ such that the discrete time system is also reactive, i.e. there is a column sum of $B$ greater than 1. However for a predetermined time-step $\tau$ there are continuous time systems $x'=Ax$ that are reactive but for which their discrete counterparts $x(t+\tau)=e^{A\tau}x(t)$ are not reactive. One such example is system (\ref{eq:asymgrow1}) with $\epsilon=0.9$ and $\tau=1$. 

Lastly, we hope our work can be used to better understand the transient dynamics in marine metapopulations for which habitat patches have already been classified as sources and sinks. For these systems the transient dynamics that may occur following a disturbance depend directly on the new distribution of the population. If the remaining population is distributed amongst sink patches then it initially declines, even if it eventually recovers. Likewise if the population is distributed amongst source patches then it initially grows, though this growth may not necessarily occur on the source patch itself. In addition, the relationship between transient dynamics and the source-sink distribution of marine metapopulations may also be useful when examining the dynamics that can occur following the protection of new marine environments, such as newly implemented Marine Protected Areas.

\appendix
\section{Two patch example with migration} 
\label{appendix:a}

Consider a different metapopulation on two patches where individuals are born only onto their patch, but can now also migrate between patches. This is the case for many terrestrial species that live on patchy landscapes, where individuals can migrate between habitat patches. The dynamics of this metapopulation is then described with the following set of ODEs:

\begin{align}
x_1'&=b_{1}x_1-m_{21}x_1+m_{12}x_2-d_1x_1\label{eq:2patch1}\\
x_2'&=b_{2}x_2-m_{12}x_2+m_{21}x_1-d_2x_2\nonumber,
\end{align}
where $b_{i}$ is the birth rate on patch $i$, $m_{ij}$ is the migration rate from patch $j$ to patch $i$, and $d_i$ is the death rate on patch $i$. The lifecycle graph for this system is shown in Figure \ref{fig:patch2}. 

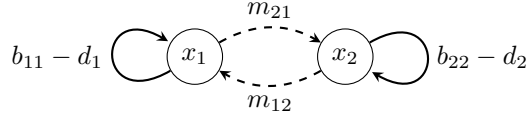
\begin{figure}
\centering
\begin{tikzpicture}
\node (x1) [stage] {$x_1$};
\node (x2) [stage, right of=x, xshift=1cm] {$x_2$};
\draw [arrow, dashed] (x1) to [out=30,in=150] node[above] {$m_{21}$} (x2);  
\draw [arrow, dashed] (x2) to [out=210,in=330] node[below] {$m_{12}$}(x1);
\draw [arrow] (x1) edge[out=210,in=150,loop] node[left] {$b_{11}-d_1$} (x1);
\draw [arrow] (x2) edge[out=30,in =330,loop] node[right] {$b_{22}-d_2$} (x2);
\end{tikzpicture}
\caption{Digraph for system (\ref{eq:2patch1}). The directed edges represent the movement of individuals from the outgoing patch to the incoming patch. The self loops are the birth rate minus the death rate on a patch.}
\label{fig:patch2}
\end{figure}

We again decompose $A=F-V$ and construct our next generation matrix, $K=FV^{-1}$, though now $V$ is not diagonal nor in the same form as required for Lemma \ref{thm:Riyi}.

\begin{align*}
A&=\begin{bmatrix}
b_{1}-m_{21}-d_1 & m_{12}\\
m_{21} & b_{2}-m_{12}-d_2
\end{bmatrix}\\
F&=\begin{bmatrix}
b_{1} & 0\\
0 & b_{2}
\end{bmatrix}\\
V&=\begin{bmatrix}
d_1+m_{21} & -m_{12}\\
-m_{21} & d_2+m_{12}
\end{bmatrix}\\
K=FV^{-1}&=\begin{bmatrix}
\frac{b_1(d_2+m_{12})}{d_1d_2+d_1m_{12}+d_2m_{21}} & \frac{b_1m_{12}}{d_1d_2+d_1m_{12}+d_2m_{21}}\\
\frac{b_2m_{21}}{d_1d_2+d_1m_{12}+d_2m_{21}} & \frac{b_2(d_1+m_{21})}{d_1d_2+d_1m_{12}+d_2m_{21}}
\end{bmatrix}
\end{align*}

The entries of $K$ may seem counter intuitive, but they represent the infinite sum of a geometric series. Consider the first entry, $k_{11}$. If we are tracking the total number of new individuals produced by one individual starting on patch 1, then this individual can either produce new offspring in patch 1 immediately, or it can migrate to patch 2, then back to patch 1 and produce offspring, or migrate again and produce more offspring. The entry $k_{11}$ is then the birth rate in patch $1$ multiplied by the residence time in patch $1$, multiplied by a geometric series where the ratio is the probability of surviving the migration from patch 1 to patch 2 and then back to patch 1. Mathematically 

\begin{equation*}
k_{11}=\frac{b_1}{m_{21}+d_1}\left[\sum_{i=0}^{\infty} \left(\frac{m_{21}}{m_{21}+d_1}\frac{m_{12}}{m_{12}+d_2}\right)^i\right]
\end{equation*}

Now consider $\lambda_1$ and $R_1$,  the measures of transient growth for patch 1: 

\begin{align*}
\lambda_1&=b_1-d_1\\
R_1&=\frac{b_1d_2+b_1m_{12}+b_2m_{21}}{d_1d_2+d_1m_{12}+d_2m_{21}}.
\end{align*}

We can see that 
\begin{equation*}
\lim_{d_2\rightarrow \infty}R_1=\frac{b_1}{d_1+m_{21}}.
\end{equation*}
Therefore even if $\lambda_1>0$, and so $b_1>d_1$, as $m_{21}$ becomes large, $R_1<1$.

In the other direction, if $d_1=2b_1$, then 
\begin{equation*}
\lim_{b_1\rightarrow 0} R_1=\frac{b_2}{d_2}.
\end{equation*}
Therefore even if $R_1>1$, we can still have $\lambda_1=b_1-d_1<0$. We present this example to demonstrate that if the assumptions of Lemma \ref{lemma:R0Ri} are not met, there is no longer a one-to-one relationship between $\lambda_j$ and $R_j$.

Now we might also consider moving the off diagonal entries of $V$ into $F$ so that $V$ becomes diagonal. This is similar to considering migrating individuals as new individuals entering a patch. In this case

\begin{align*}
F&=\begin{bmatrix}
b_{1} & m_{12}\\
m_{21} & b_{2}
\end{bmatrix}\\
V&=\begin{bmatrix}
d_1+m_{21} & 0\\
0 & d_2+m_{12}
\end{bmatrix}\\
K=FV^{-1}&=\begin{bmatrix}
\frac{b_1}{d_1+m_{21}} & \frac{m_{12}}{d_2+m_{12}}\\
\frac{m_{21}}{d_1+m_{21}} & \frac{b_2}{d_2+m_{12}}
\end{bmatrix}
\end{align*}

Here $\lambda_1=b_1-d_1>0$ is equivalent to $b_1>d_1$, which is then equivalent to $R_1=(b_1+m_{21})/(d_1+m_{21})>1$. However, in this case $R_1$ no longer tracks the total number of new individuals produced on all patches over one generation. This new $R_1$ could perhaps be interpreted as the total number of new individuals produced on patch 1 by a single individual on patch 1 before that individual dies or migrates, plus the probability that the individual migrates to patch 2 before it dies. However, this will no longer be a    biologically useful measure of a source or a sink.

\section*{Acknowledgments}
The authors would like to thank the members of the Lewis Lab for many helpful discussions and suggestions. PDH gratefully acknowledges an NSERC-CGSM scholarship, Queen Elizabeth II scholarship, and Alberta Graduate Excellence Scholarship; MAL gratefully acknowledges the Canada Research Chair program and an NSERC Discovery grant; and PvdD gratefully acknowledges an NSERC Discovery grant.

\bibliographystyle{siamplain}
\bibliography{C:/Users/Peter/Documents/Research/TransientPaper/transient}

\begin{thebibliography}{10}

\bibitem{Adams2015}
{\sc T.~P. Adams, R.~Proud, and K.~D. Black}, {\em {Connected networks of sea
  lice populations: dynamics and implications for control}}, Aquac. Environ.
  Interact., 6 (2015), pp.~273--284, \url{https://doi.org/10.3354/aei00133}.

\bibitem{Arino2003}
{\sc J.~Arino and P.~van~den Driessche}, {\em {A multi-city epidemic model}},
  Math. Popul. Stud., 10 (2003), pp.~175--193,
  \url{https://doi.org/10.1080/08898480306720}.

\bibitem{Armsworth2002}
{\sc P.~R. Armsworth}, {\em {Recruitment limitation, population regulation, and
  larval connectivity in reef fish metapopulations}}, Ecology, 83,
  pp.~1092--1104,
  \url{https://doi.org/10.1890/0012-9658(2002)083[1092:RLPRAL]2.0.CO;2}.

\bibitem{Botsford1994}
{\sc L.~W. Botsford, C.~L. Moloney, A.~Hastings, J.~L. Largier, T.~M. Powell,
  K.~Higgins, and J.~F. Quinn}, {\em {The influence of spatially and temporally
  varying oceanographic conditions on meroplanktonic metapopulations}}, Deep.
  Res. Part II, 41 (1994), pp.~107--145,
  \url{https://doi.org/10.1016/0967-0645(94)90064-7}.

\bibitem{Caswell2005}
{\sc H.~Caswell and M.~G. Neubert}, {\em {Reactivity and transient dynamics of
  discrete-time ecological systems}}, J. Differ. Equations Appl., 11 (2005),
  pp.~295--310, \url{https://doi.org/10.1080/10236190412331335382}.

\bibitem{Cowen2006}
{\sc R.~K. Cowen, C.~B. Paris, and A.~Srinivasan}, {\em {Scaling of
  connectivity in marine populations}}, Science, 311 (2006), pp.~522--527,
  \url{https://doi.org/10.1126/science.1122039}.

\bibitem{Cowen2009}
{\sc R.~K. Cowen and S.~Sponaugle}, {\em {Larval dispersal and marine
  population connectivity}}, Ann. Rev. Mar. Sci., 1 (2009), pp.~443--466,
  \url{https://doi.org/10.1146/annurev.marine.010908.163757}.

\bibitem{Cushing2016}
{\sc J.~M. Cushing and O.~Diekmann}, {\em {The many guises of $R_0$ (a didactic
  note)}}, J. Theor. Biol., 404 (2016), pp.~295--302,
  \url{https://doi.org/10.1016/j.jtbi.2016.06.017}.

\bibitem{Figueira2006}
{\sc W.~F. Figueira and L.~B. Crowder}, {\em {Defining patch contribution in
  source-sink metapopulations: the importance of including dispersal and its
  relevance to marine systems}}, Popul. Ecol., 48 (2006), pp.~215--224,
  \url{https://doi.org/10.1007/s10144-006-0265-0}.

\bibitem{FitzHugh1961}
{\sc R.~FitzHugh}, {\em {Impulses and physiological states in theoretical
  models of nerve membrane}}, Biophys. J., 1 (1961), pp.~445--466,
  \url{https://doi.org/10.1016/S0006-3495(61)86902-6}.

\bibitem{Gyllenberg1997b}
{\sc M.~Gyllenberg, A.~Hastings, and I.~Hanski}, {\em {5 - Structured
  metapopulation models}}, in Metapopulation Biology, I.~Hanski and M.~E.
  Gilpin, eds., Academic Press, San Diego, 1997, pp.~93--122,
  \url{https://doi.org/10.1016/B978-012323445-2/50008-0}.

\bibitem{Hanski1994}
{\sc I.~Hanski and C.~D. Thomas}, {\em {Metapopulation dynamics and
  conservation: a spatially explicit model applied to butterflies}}, Biol.
  Conserv., 68 (1994), pp.~167--180,
  \url{https://doi.org/10.1016/0006-3207(94)90348-4}.

\bibitem{Harrington2020}
{\sc P.~D. Harrington and M.~A. Lewis}, {\em {A next-generation approach to
  calculate source–sink dynamics in marine metapopulations}}, Bull. Math.
  Biol., 82 (2020), pp.~1--44,
  \url{https://doi.org/10.1007/s11538-019-00674-1}.

\bibitem{Hastings2001}
{\sc A.~Hastings}, {\em {Transient dynamics and persistence of ecological
  systems}}, Ecol. Lett., 4 (2001), pp.~215--220,
  \url{https://doi.org/10.1046/j.1461-0248.2001.00220.x}.

\bibitem{Hastings2004}
{\sc A.~Hastings}, {\em {Transients: the key to long-term ecological
  understanding?}}, Trends Ecol. Evol., 19 (2004), pp.~39--45,
  \url{https://doi.org/10.1016/j.tree.2003.09.007}.

\bibitem{Hastings2018}
{\sc A.~Hastings, K.~C. Abbott, K.~Cuddington, T.~Francis, G.~Gellner, Y.~C.
  Lai, A.~Morozov, S.~Petrovskii, K.~Scranton, and M.~L. Zeeman}, {\em
  {Transient phenomena in ecology}}, Science, 361 (2018),
  \url{https://doi.org/10.1126/science.aat6412}.

\bibitem{Hastings1994}
{\sc A.~Hastings and K.~Higgins}, {\em {Persistence of transients in spatially
  structured ecological models}}, Science, 263 (1994), pp.~1133--1136,
  \url{https://doi.org/10.1126/science.263.5150.1133}.

\bibitem{Higgins1997}
{\sc K.~Higgins, A.~Hastings, J.~N. Sarvela, and L.~W. Botsford}, {\em
  {Stochastic dynamics and deterministic skeletons: population behavior of
  Dungeness crab}}, Science, 276 (1997), pp.~1431--1434,
  \url{https://doi.org/10.1126/science.276.5317.1431}.

\bibitem{Hillen2015}
{\sc T.~Hillen, B.~Greese, J.~Martin, and G.~de~Vries}, {\em {Birth-jump
  processes and application to forest fire spotting}}, J. Biol. Dyn., 9 (2015),
  pp.~104--127, \url{https://doi.org/10.1080/17513758.2014.950184}.

\bibitem{Horn2012}
{\sc R.~A. Horn and C.~R. Johnson}, {\em {Matrix Analysis}}, Cambridge
  University Press, New York, NY, USA, 2nd~ed., 2012.

\bibitem{Huang2016b}
{\sc Q.~Huang, Y.~Jin, and M.~A. Lewis}, {\em {$R_0$ analysis of a
  benthic-drift model for a stream population}}, SIAM J. Appl. Dyn. Syst., 15
  (2016), pp.~287--321, \url{https://doi.org/10.1137/15M1014486}.

\bibitem{Huang2015}
{\sc Q.~Huang and M.~A. Lewis}, {\em {Homing fidelity and reproductive rate for
  migratory populations}}, Theor. Ecol., 8 (2015), pp.~187--205,
  \url{https://doi.org/10.1007/s12080-014-0243-7}.

\bibitem{Husband1996}
{\sc B.~C. Husband and S.~C.~H. Barrett}, {\em {A metapopulation perspective in
  plant population biology}}, J. Ecol., 84 (1996), pp.~461--469,
  \url{https://doi.org/10.2307/2261207}.

\bibitem{Jones2009}
{\sc G.~Jones, G.~Almany, G.~Russ, P.~Sale, R.~Steneck, M.~van Oppen, and
  B.~Willis}, {\em {Larval retention and connectivity among populations of
  corals and reef fishes: history, advances and challenges}}, Coral Reefs, 28
  (2009), pp.~307--325, \url{https://doi.org/10.1007/s00338-009-0469-9}.

\bibitem{Krkosek2010a}
{\sc M.~Krko{\v{s}}ek and M.~A. Lewis}, {\em {An $R_0$ theory for source-sink
  dynamics with application to Dreissena competition}}, Theor. Ecol., 3 (2010),
  pp.~25--43, \url{https://doi.org/10.1007/s12080-009-0051-7}.

\bibitem{Levins1969}
{\sc R.~Levins}, {\em {Some demographic and genetic consequences of
  environmental heterogeneity for biological control}}, Bull. Entomol. Soc.
  Am., 15 (1969), pp.~237--250, \url{https://doi.org/10.1088/2058-9565/aaba1a}.

\bibitem{lloyd1996}
{\sc A.~L. Lloyd and R.~M. May}, {\em {Spatial heterogeneity in epidemic
  models}}, J. Theor. Biol., 179 (1996), pp.~1--11,
  \url{https://doi.org/10.1006/jtbi.1996.0042}.

\bibitem{Ludwig1978}
{\sc D.~Ludwig, D.~D. Jones, and C.~S. Holling}, {\em {Qualitative analysis of
  insect outbreak systems: the spruce budworm and forest}}, J. Anim. Ecol., 47
  (1978), pp.~315--332, \url{https://doi.org/10.2307/3939}.

\bibitem{Lutscher2020}
{\sc F.~Lutscher and X.~Wang}, {\em {Reactivity of communities at equilibrium
  and periodic orbits}}, J. Theor. Biol., 493 (2020), p.~110240,
  \url{https://doi.org/10.1016/j.jtbi.2020.110240}.

\bibitem{Mari2017}
{\sc L.~Mari, R.~Casagrandi, A.~Rinaldo, and M.~Gatto}, {\em {A generalized
  definition of reactivity for ecological systems and the problem of transient
  species dynamics}}, Methods Ecol. Evol., 8 (2017), pp.~1574--1584,
  \url{https://doi.org/10.1111/2041-210X.12805}.

\bibitem{Mckenzie2012}
{\sc H.~W. Mckenzie, Y.~Jin, J.~Jacobsen, and M.~A. Lewis}, {\em {$R_0$
  analysis of a spatiotemporal model for a stream population}}, SIAM J. Appl.
  Dyn. Syst., 11 (2012), pp.~567--596, \url{https://doi.org/10.1137/100802189}.

\bibitem{Morozov2020}
{\sc A.~Morozov, K.~Abbott, K.~Cuddington, T.~Francis, G.~Gellner, A.~Hastings,
  Y.~C. Lai, S.~Petrovskii, K.~Scranton, and M.~L. Zeeman}, {\em {Long
  transients in ecology: theory and applications}}, Phys. Life Rev., 32 (2020),
  pp.~1--40, \url{https://doi.org/10.1016/j.plrev.2019.09.004}.

\bibitem{Morris1963}
{\sc R.~F. Morris}, {\em {The dynamics of epidemic spruce budworm
  populations}}, Mem. Entomol. Soc. Canada, 95 (1963), pp.~1--12,
  \url{https://doi.org/10.4039/entm9531fv}.

\bibitem{Nagumo1962}
{\sc J.~Nagumo, S.~Arimoto, and S.~Yoshizawa}, {\em {An active pulse
  transmission line simulating nerve axon}}, Proc. IRE, 50 (1962),
  pp.~2061--2070, \url{https://doi.org/10.1109/JRPROC.1962.288235}.

\bibitem{Caswell1997}
{\sc M.~G. Neubert and H.~Caswell}, {\em {Alternatives to resilience for
  measuring the responses of ecological systems to perturbations}}, Ecology, 78
  (1997), pp.~653--665,
  \url{https://doi.org/10.1890/0012-9658(1997)078[0653:ATRFMT]2.0.CO;2}.

\bibitem{Neubert2002}
{\sc M.~G. Neubert, H.~Caswell, and J.~D. Murray}, {\em {Transient dynamics and
  pattern formation: reactivity is necessary for Turing instabilities}}, Math.
  Biosci., 175 (2002), pp.~1--11,
  \url{https://doi.org/10.1016/S0025-5564(01)00087-6}.

\bibitem{Neubert2004}
{\sc M.~G. Neubert, T.~Klanjscek, and H.~Caswell}, {\em {Reactivity and
  transient dynamics of predator-prey and food web models}}, Ecol. Modell., 179
  (2004), pp.~29--38, \url{https://doi.org/10.1016/j.ecolmodel.2004.05.001}.

\bibitem{Noschese2013}
{\sc S.~Noschese, L.~Pasquini, and L.~Reichel}, {\em {Tridiagonal Toeplitz
  matrices: properties and novel applications}}, Numer. Linear Algebra Appl.,
  20 (2013), pp.~302--326, \url{https://doi.org/10.1002/nla.1811}.

\bibitem{Pulliam1988}
{\sc H.~R. Pulliam}, {\em {Sources, sinks and population regulation}}, Am.
  Nat., 132 (1988), pp.~652--661, \url{https://doi.org/10.1086/284880}.

\bibitem{Roughgarden1988}
{\sc J.~Roughgarden, S.~Gaines, and H.~Possingham}, {\em {Recruitment dynamics
  in complex life cycles}}, Science, 241 (1988), pp.~1460--1466,
  \url{https://doi.org/10.1126/science.11538249}.

\bibitem{Safranyik2007}
{\sc L.~Safranyik and A.~L. Carroll}, {\em {The biology and epidemiology of the
  mountain pine beetle in lodgepole pine forests.}}, Canadian Forest Service,
  Victoria, Canada, 2007, pp.~3--66.

\bibitem{Snyder2010}
{\sc R.~E. Snyder}, {\em {What makes ecological systems reactive?}}, Theor.
  Popul. Biol., 77 (2010), pp.~243--249,
  \url{https://doi.org/10.1016/j.tpb.2010.03.004}.

\bibitem{Stott2010}
{\sc I.~Stott, M.~Franco, D.~Carslake, S.~Townley, and D.~Hodgson}, {\em {Boom
  or bust? A comparative analysis of transient population dynamics in plants}},
  J. Ecol., 98 (2010), pp.~302--311,
  \url{https://doi.org/10.1111/j.1365-2745.2009.01632.x}.

\bibitem{Stott2011}
{\sc I.~Stott, S.~Townley, and D.~J. Hodgson}, {\em {A framework for studying
  transient dynamics of population projection matrix models}}, Ecol. Lett., 14
  (2011), pp.~959--970, \url{https://doi.org/10.1111/j.1461-0248.2011.01659.x}.

\bibitem{Stott2011a}
{\sc I.~Stott, S.~Townley, and D.~J. Hodgson}, {\em {A framework for studying
  transient dynamics of population projection matrix models}}, Ecol. Lett., 14
  (2011), pp.~959--970, \url{https://doi.org/10.1111/j.1461-0248.2011.01659.x}.

\bibitem{Thieme2003}
{\sc H.~R. Thieme}, {\em {Mathematics in Population Biology}}, Princeton
  University Press, 2003.

\bibitem{Thieme2009}
{\sc H.~R. Thieme}, {\em {Spectral bound and reproduction number for
  infinite-dimensional population structure and time heterogeneity}}, SIAM J.
  Appl. Math., 70 (2009), pp.~188--211,
  \url{https://doi.org/10.1137/080732870}.

\bibitem{Townley2007}
{\sc S.~Townley, D.~Carslake, O.~Kellie-Smith, D.~Mccarthy, and D.~Hodgson},
  {\em {Predicting transient amplification in perturbed ecological systems}},
  J. Appl. Ecol., 44 (2007), pp.~1243--1251,
  \url{https://doi.org/10.1111/j.1365-2664.2007.01333.x}.

\bibitem{Townley2008}
{\sc S.~Townley and D.~J. Hodgson}, {\em {Erratum et addendum: Transient
  amplification and attenuation in stage-structured population dynamics}}, J.
  Appl. Ecol., 45 (2008), pp.~1836--1839,
  \url{https://doi.org/10.1111/j.1365-2664.2008.01562.x}.

\bibitem{Hogben2006}
{\sc M.~Tsatsomeros}, {\em {Matrix equalities and inequalities}}, in Handbook
  of Linear Algebra, L.~Hogben, ed., CRC Press, 2006, ch.~14.

\bibitem{Verdy2008}
{\sc A.~Verdy and H.~Caswell}, {\em {Sensitivity analysis of reactive
  ecological dynamics}}, Bull. Math. Biol., 70 (2008), pp.~1634--1659,
  \url{https://doi.org/10.1007/s11538-008-9312-7}.

\bibitem{Williams2000}
{\sc D.~W. Williams and A.~M. Liebhold}, {\em {Spatial synchrony of spruce
  budworm outbreaks in eastern North America}}, Ecology, 81 (2000),
  pp.~2753--2766,
  \url{https://doi.org/10.1890/0012-9658(2000)081[2753:SSOSBO]2.0.CO;2}.

\end{thebibliography}

\end{document}